\documentclass{article}
\usepackage{graphicx}
\usepackage{amsmath,amssymb,mathtools}
\usepackage{booktabs}
\usepackage{makecell}
\usepackage{url}
\usepackage{mathptmx}
\usepackage{times}
\usepackage{bm}
\makeatletter
\let\cl@part\relax  
\makeatother
\usepackage{algorithm}        
\usepackage{algpseudocode}    
\makeatletter
\@ifpackageloaded{algpseudocode}{%
  \providecommand{\STATE}{\State}
  \providecommand{\FOR}{\For}
  \providecommand{\ENDFOR}{\EndFor}

}{}
\makeatother
\newenvironment{proof}[1][Proof]{\par\noindent\textbf{#1.}\ }{\hfill$\square$\par}
\usepackage[numbers,sort&compress]{natbib}
\usepackage{aliascnt}
\makeatletter
\def\@begintheorem#1#2{%
  \trivlist
  \item[\hskip\labelsep{\bfseries #1\ #2}]\itshape}
\def\@opargbegintheorem#1#2#3{%
  \trivlist
  \item[\hskip\labelsep{\bfseries #1\ #2\ (#3)}]\itshape}
\makeatother
\newtheorem{theorem}{Theorem}

\newaliascnt{lemma}{theorem}
\newtheorem{lemma}[lemma]{Lemma}
\aliascntresetthe{lemma}

\newaliascnt{proposition}{theorem}
\newtheorem{proposition}[proposition]{Proposition}
\aliascntresetthe{proposition}

\newaliascnt{corollary}{theorem}
\newtheorem{corollary}[corollary]{Corollary}
\aliascntresetthe{corollary}

\newaliascnt{conjecture}{theorem}

\aliascntresetthe{conjecture}

\newaliascnt{definition}{theorem}

\aliascntresetthe{definition}

\newaliascnt{remark}{theorem}
\newtheorem{remark}[remark]{Remark}
\aliascntresetthe{remark}

\newaliascnt{assumption}{theorem}
\newtheorem{assumption}[assumption]{Assumption}
\aliascntresetthe{assumption}

\newaliascnt{hypothesis}{theorem}

\aliascntresetthe{hypothesis}

\newaliascnt{property}{theorem}

\aliascntresetthe{property}

\usepackage[capitalize]{cleveref}
\Crefname{figure}{Fig.}{Figs.}
\crefname{lemma}{Lemma}{Lemmas}
\Crefname{lemma}{Lemma}{Lemmas}
\crefname{proposition}{Proposition}{Propositions}
\Crefname{proposition}{Proposition}{Propositions}
\crefname{corollary}{Corollary}{Corollaries}
\Crefname{corollary}{Corollary}{Corollaries}
\crefname{definition}{Definition}{Definitions}
\Crefname{definition}{Definition}{Definitions}
\crefname{remark}{Remark}{Remarks}
\Crefname{remark}{Remark}{Remarks}
\crefname{assumption}{Assumption}{Assumptions}
\Crefname{assumption}{Assumption}{Assumptions}
\crefname{conjecture}{Conjecture}{Conjectures}
\Crefname{conjecture}{Conjecture}{Conjectures}
\crefname{hypothesis}{Hypothesis}{Hypotheses}
\Crefname{hypothesis}{Hypothesis}{Hypotheses}
\crefname{property}{Property}{Properties}
\Crefname{property}{Property}{Properties}
\let\half\undefined
\usepackage{mpcsymbols}

\crefformat{equation}{(#2#1#3)}
\providecommand{\runauthor}[1]{}%
\providecommand{\runtitle}[1]{}%
\providecommand{\thanksref}[1]{}%

\begin{document}

\title{Disturbance Attenuation Regulator II: Stage Bound Finite Horizon Solution\thanks{The authors gratefully acknowledge the financial support of the National Science Foundation (NSF) under Grant Nos. 2027091 and 2138985.}}

\author{Davide Mannini\thanks{Department of Chemical Engineering, University of California, Santa Barbara. Email: dmannini@ucsb.edu} \and James B. Rawlings\thanks{Department of Chemical Engineering, University of California, Santa Barbara. Email: jbraw@ucsb.edu}}

\maketitle

\begin{abstract}
This paper develops a generalized finite horizon recursive solution to the discrete time stage bound disturbance attenuation regulator (StDAR) for state feedback control. This problem addresses linear dynamical systems subject to stage bound disturbances, i.e., disturbance sequences constrained independently at each time step through stagewise squared two-norm bounds. The term generalized indicates that the results accommodate arbitrary initial states. By combining game theory and dynamic programming, this work derives a recursive solution for the optimal state feedback policy. The optimal policy is nonlinear in the state and requires solving a tractable convex optimization for the Lagrange multiplier vector at each stage; the control is then explicit. For systems with constant stage bound, we introduce a steady-state StDAR whose solution reduces to a tractable linear matrix inequality (LMI) with empirical computational cost approximately cubic in $n$. Numerical examples illustrate the properties of the solution.

This work provides a complete feedback solution to the StDAR for arbitrary initial states. Companion papers address the signal bound disturbance attenuation regulator (SiDAR): the finite horizon solution in Part~I-A and convergence properties in Part~I-B.
\end{abstract}

\section{Introduction}
\label{sec:intro}

The disturbance attenuation regulator (DAR), alternatively termed the sensitivity minimization problem, formulates robust control as a deterministic minmax game between controller and disturbance. In this sequential dynamic noncooperative zero-sum game (a Stackelberg game), the disturbance optimizes first (follower) and the control optimizes second (leader). The control objective is to maintain low cost despite any admissible bounded disturbance. These problems present fundamental mathematical challenges. Witsenhausen \citep{witsenhausen:1968} noted that only weak duality holds in relevant cases. Moreover, global optima need not satisfy stationarity conditions, i.e., points in the domain of a function at which the gradient is zero, and gradient based algorithms may fail because search domain restrictions can inadvertently eliminate solution branches, blocking convergence even when portions of the solution lie within the search region.

Two constraints bound disturbances differently. The signal bound disturbance attenuation regulator (SiDAR) constrains total energy through a single squared signal two-norm bound over all time steps. The stage bound disturbance attenuation regulator (StDAR) independently constrains energy at each time step through stagewise squared two-norm bound. These differences produce distinct computational requirements and optimal policies.

The StDAR addresses a fundamental limitation of the SiDAR. The aggregate signal bound constraint imposes temporal coupling: the admissible disturbance at each stage depends on all past disturbance realizations through the remaining budget. This coupling treats the disturbance as an adversarial agent that strategically allocates energy from a fixed budget to maximize system cost. While this game-theoretic perspective provides valuable theoretical insights, it is less representative of physical disturbances in practical applications. Actual disturbances arising from environmental conditions, measurement noise, or model uncertainty do not possess knowledge of past realizations, nor do they strategically coordinate to deplete an energy budget. Stage bound constraints eliminate this artificial coupling by independently bounding disturbance magnitude at each stage. This problem naturally accommodates persistent disturbances that act consistently over time or time-varying disturbances whose bound change with operating conditions, without requiring the disturbance to implement a coordinated strategy across the horizon. Consequently, the StDAR seems more suitable for designing controllers robust to the types of uncertain disturbances encountered in practical control applications.

Bulgakov \citep{bulgakov:1946} established early foundations for stage bound methods in the 1940s by determining maximum terminal deviations under stagewise input constraints. Soviet/Russian researchers developed this work extensively (it remains known as \textit{Bulgakov's problem}) yet Western research remained largely disconnected from these developments.

Game-theoretic DAR problems emerged in the 1960s. Stikhin \citep{stikhin:1963} and Gadzhiev \citep{gadzhiev:1962} independently addressed stage bound and signal bound cases respectively, with Gadzhiev deriving nonlinear optimal policies for linear systems. Following Dorato and Drenick's \citep{dorato:drenick:1966} introduction of these concepts to Western audiences, extensive research ensued throughout the 1960s-1970s \citep{koivuniemi:1966,ragade:sarma:1967,salmon:1968,rhodes:luenberger:1969,kimura:1970,medanic:andjelic:1971,ulanov:1971,bertsekas:rhodes:1973,yakubovich:1975,barabanov:granichin:1984}. Progress stalled as weak duality and nonstationary global solutions prevented complete solutions for either constraint paradigm.

Frequency domain reformulations revitalized research in the 1980s. Zames \citep{zames:1981} defined a frequency domain optimization problem that proved equivalent to the time domain SiDAR, introducing the attenuation level $\gamma$ (a Lagrange multiplier analogue). This $H_{\infty}$ framework shifted Western research attention predominantly toward signal bound problems, though stage bound problems continued to receive attention in Soviet and Russian literature. Glover and Doyle \citep{glover:doyle:1988} developed corresponding time domain dual Riccati recursions for continuous time systems at the origin for signal bound disturbances. Basar \citep{basar:1989a} provided finite and infinite horizon dynamic game problems for the signal bound case. Vidyasagar \citep{vidyasagar:1986} extended the time domain $H_{\infty}$ theory to stage bound disturbances, reintroducing this problem to Western literature. These advances assumed zero initial state, however, disconnecting them from earlier game-theoretic treatments.

Dorato \citep{dorato:1987} and Khlebnikov, Polyak, and Kuntsevich \citep{khlebnikov:polyak:kuntsevich:2011} provide historical perspectives emphasizing Western and Soviet/Russian contributions respectively.

Currently, no direct feedback solution exists for the StDAR that is valid throughout the entire state space. Existing methods derive solutions only at the origin. This limitation has practical consequences: large disturbances or setpoint changes drive systems away from the origin where existing solutions become suboptimal.

We derive a generalized, finite horizon, recursive feedback solution to the StDAR with:
\begin{itemize}
\item validity throughout the state space, deriving the optimal state feedback policy for any state
\item optimal state feedback policies that are nonlinear in the state and require solving a tractable convex optimization for the Lagrange multiplier vector at each stage; the control gain is then explicit
\end{itemize}

\Cref{sec:setup} formulates the finite horizon StDAR. \Cref{sec:stage} develops the recursive solution via backward dynamic programming. \Cref{sec:ss-stage} introduces the steady-state StDAR for systems with constant stage bound and derives its tractable LMI, whose empirical computational cost is approximately cubic in $n$. \Cref{sec:num-example} illustrates the theory with numerical examples, and \Cref{sec:end} closes the paper with a summary of the main findings. The appendix compiles fundamental results used throughout the paper.

Companion papers address the finite horizon SiDAR solution \citep{mannini:rawlings:2026a} and the steady-state problem and convergence properties for the SiDAR solution \citep{mannini:rawlings:2026b}.

\textit{Notation:} Let $\bbR$ denote the reals and $\bbI$ the integers. $\mathbb{R}^{m \times n}$ denotes the space of $m \times n$ real matrices and $\mathbb{S}^n$ denotes the set of $n \times n$ real symmetric positive definite matrices. The \(\norm{x}\) denotes the Euclidean norm of vector \(x\) and \(\norm{M}\) denotes the induced 2-norm of matrix \(M\). For a vector $w \in \mathbb{R}^p$, let $\wseq$ denote a sequence $\wseq \eqbyd (w_0, w_1, \dots, w_{N-1})$ over a finite horizon $N$. The two-norm of a signal $\wseq$ is defined as $\smax{\wseq} \eqbyd ( \sum_{k=0}^{N-1} \norm{w_k}^2 )^{1/2}$. The column space (range) and nullspace of a matrix $M$ are denoted by $\mathcal{R}(M)$ and $\mathcal{N}(M)$, respectively. The pseudoinverse of a matrix $M$ is denoted as $M^{\dagger}$. For symmetric matrices $A$ and $B$, $A \succeq B$ (respectively $A \preceq B$) denotes $A - B$ is positive semidefinite (respectively negative semidefinite).

\section{StDAR Set Up}
\label{sec:setup}
Consider the following discrete time system
\begin{equation}
    x^+ = Ax + Bu + Gw \label{system}
\end{equation}  
in which $x \in \bbR^n$ is the state, $u \in \bbR^m$ is the control, $w \in \bbW \subset \bbR^q$ is a disturbance, and $x^+$ is the successor state. Denote the horizon length i.e., number of time steps in the future, as $N \in \mathbb{I}_{\ge 1}$. Define the control and disturbance sequences: $\useq \eqbyd (u_0,u_1,\dots,u_{N-1})$, $\wseq \eqbyd (w_0,w_1,\dots,w_{N-1})$. Consider the following stage bound disturbance constraint set (stagewise two-norm bound)
\[
	\bbW_{st}
	\eqbyd
	\Bigl\{
	\wseq
	 \mid
	|w_k|^2 \leq \alpha_k
	\Bigr\}
	\qquad k \in[0,1,\dots,N-1]\] 
where $\alpha_k > 0$ for all $k$. Define the following objective function
\begin{equation}
    V(x_0, \useq, \wseq) = \sum_{k=0}^{N-1} \ell(x_k, u_k)  + \ell_f(x_N) \label{maincost}
\end{equation}
where $x_0$ is the initial state, $\ell(\cdot)$ the stage cost, $\ell_f(\cdot)$ the terminal cost
\[
    \ell(x,u) = (1/2)x'Qx + (1/2)u'Ru \qquad \ell_f(x) = (1/2)x'P_fx
\]
in which $Q \succeq 0$, $R \succ 0$, and $P_f \succeq 0$. The state $x_k$ in \eqref{maincost} denotes the solution of \eqref{system} at stage $k$ generated from the initial state $x_0$ by the control and disturbance sequences $\useq$ and $\wseq$. We state the following assumptions.
\begin{assumption}
For the linear system \eqref{system}, $(A,B)$ stabilizable and  $(A,Q)$ detectable. \label{asst1}
\end{assumption}
\begin{assumption}
$\mathcal{R}(G)\subseteq\mathcal{R}(B)$. \label{asst2}
\end{assumption}
\begin{assumption}
$G'P_fG\neq0$. \label{asst3}
\end{assumption}
\begin{assumption}
$Q\succ0$, $P_f \succ0$. \label{asst4}
\end{assumption}

Let $\overline{\alpha} \eqbyd \sum_{k=0}^{N-1}\alpha_k > 0$. The \textit{stage bound disturbance attenuation regulator} (StDAR) is obtained from the optimization problem
\begin{equation}
    V^*(x_0) \eqbyd \min_{u_0}\max_{w_0} \; \min_{u_1}\max_{w_1} \; \cdots \min_{u_{N-1}}\max_{w_{N-1}} \;
 \frac{V(x_0, \useq, \wseq)}{\overline{\alpha}} \; \; \wseq \in \bbW_{st} \label{stagedp-w}
\end{equation}
subject to \eqref{system}. The optimization \eqref{stagedp-w} is well-defined since $V$ is continuous on the compact set $\bbW_{st}$. Under \cref{asst3}, the maximum over the disturbance is attained on the stagewise boundary $|w_k|^2 = \alpha_k$ for each $k$, as shown in the proofs of \cref{prop:2stage-StDAR,prop:ndstage}. At this boundary $\sum_{k=0}^{N-1}|w_k|^2 = \overline{\alpha}$, so the normalized cost $V^*(x_0)$ equals the attenuation ratio $V(x_0,\useq,\wseq) / \sum_{k=0}^{N-1}|w_k|^2$ at the worst case.

\begin{remark}[Relation to classical control objectives]
The StDAR objective \eqref{stagedp-w} measures the worst-case time-averaged energy of the regulated state and control under adversarial stagewise disturbance bounds $|w_k|^2 \leq \alpha_k$. This differs from $H_\infty$ control, which measures the worst-case state energy under a single disturbance energy bound over the horizon, and from $H_2$ control, which measures expected output power under a stochastic disturbance. The averaging by $\overline{\alpha}$ is essential: a uniformly bounded disturbance can inject unbounded total energy into the linear system over an infinite horizon. The stagewise adversarial framing and resulting nonlinear optimal state feedback are reminiscent of persistent bounded disturbance control \citep{vidyasagar:1986,dahleh:diazbobillo:1995}, but the StDAR uses quadratic energies for both the disturbance bound and the regulated cost rather than maximum magnitudes.
\end{remark}

\subsection{Dynamic Programming}
The StDAR admits a standard Bellman recursion due to the stagewise separable constraint $|w_k|^2 \leq \alpha_k$. Throughout this section, we work with the equality constraint $|w|^2 = \alpha_k$, which is justified by the boundary activation arguments in \cref{prop:2stage-StDAR,prop:ndstage}. This enables the constant denominator $\overline{\alpha} = \sum_{k=0}^{N-1} \alpha_k$.

Define the value function $V_k: \bbR^n \to \bbR$ satisfying
\begin{equation}
V_k(x) = \min_u \max_{|w|^2 = \alpha_k} \left[ \frac{1}{\overline{\alpha}}\ell(x,u) + V_{k+1}(Ax+Bu+Gw) \right] \label{bellman-recursion}
\end{equation}
for $k \in \{0,\ldots,N-1\}$, where $\overline{\alpha} = \sum_{k=0}^{N-1} \alpha_k$, and with the boundary condition at $k=N$ given by
\begin{equation*}
V_N(x) = \frac{1}{\overline{\alpha}} \ell_f(x)
\end{equation*}
The optimal control policy at stage $k$ is
\begin{equation}
u_k^*(x) = \arg\min_u \max_{|w|^2 = \alpha_k} \left[ \frac{1}{\overline{\alpha}}\ell(x,u) + V_{k+1}(Ax+Bu+Gw) \right] \label{opt-control-policy}
\end{equation}
Substituting $u_k^*(x)$ into \eqref{bellman-recursion} yields
\begin{equation*}
V_k(x) = \max_{|w|^2 = \alpha_k} \left[ \frac{1}{\overline{\alpha}}\ell(x,u_k^*(x)) + V_{k+1}(Ax+Bu_k^*(x)+Gw) \right] 
\end{equation*}
and the optimal disturbance policy is
\begin{equation}
w_k^*(x) = \arg\max_{|w|^2 = \alpha_k} \left[ \frac{1}{\overline{\alpha}}\ell(x,u_k^*(x)) + V_{k+1}(Ax+Bu_k^*(x)+Gw) \right] \label{opt-dist-policy}
\end{equation}
The inner maximization in \eqref{bellman-recursion} is a constrained quadratic optimization over the compact set $\bbW_k = \{w : |w|^2 = \alpha_k\}$. We evaluate this maximization in closed form by introducing Lagrange multipliers $\lambda_k$ for the stagewise constraints, which yields the Riccati recursions and multiplier optimizations established in \cref{sec:stage}.

\section{StDAR Solution}
\label{sec:stage}

\subsection{Two-stage Solution}
We solve the two-stage version of the StDAR \eqref{stagedp-w} for the linear system \eqref{system}. The two-stage problem demonstrates how the Lagrange multipliers $(\lambda_0, \lambda_1)$ are introduced at each stage to evaluate the constrained maximizations in the Bellman recursion \eqref{bellman-recursion} in closed form.

A two-stage StDAR is
\begin{equation}
    V^*(x_0) \eqbyd \min_{u_0}\max_{w_0} \; \min_{u_1}\max_{w_1} \;
 \frac{V(x_0, \useq, \wseq)}{\overline{\alpha}} \quad \wseq \in \bbW_{st} \label{2stagedp-w}
\end{equation}
where $\useq \eqbyd (u_0,u_1)$, $\wseq \eqbyd (w_0,w_1)$, $\overline{\alpha} \eqbyd \alpha_0 + \alpha_1$, and the stage bound constraint set is
\[
\bbW_{st} \eqbyd \Bigl\{\wseq \mid |w_k|^2 \leq \alpha_k, \ k \in \{0,1\}\Bigr\}
\]
The objective function is
\[
V(x_0, \useq, \wseq) = (1/2)\bigg( x_0'Qx_0  + u_0'Ru_0 + x_1'Qx_1 + u_1'Ru_1 + x_2'P_fx_2\bigg)
\]

\Cref{prop:2stage-StDAR} states the resulting solution: the optimal multipliers $(\lambda_0^*, \lambda_1^*)$ minimize a convex value function over a feasibility domain $\Lambda_2$ that the proposition also constructs, and the optimal control and disturbance policies follow from the stationary conditions of the stacked Lagrangian evaluated at $(\lambda_0^*, \lambda_1^*)$.

\begin{proposition}[Two-stage StDAR]
\label{prop:2stage-StDAR}
Let Assumptions~1--3 hold; the matrices $M_1(\lambda_1)$ and $\Pi_1(\lambda_1)$ used below are defined later in this proposition. Define
\[
\Lambda_2 \eqbyd \Big\{(\lambda_0,\lambda_1)\in\bbR^2:
\ \lambda_1\ge\norm{G'P_fG},\
\lambda_0\ge\norm{G'\Pi_1(\lambda_1)G}\Big\}
\]
\[
\overline{\alpha} =\alpha_0 + \alpha_1
\]
Consider the convex optimization \eqref{eq:2stage-opt}
\begin{equation}
\min_{(\lambda_0,\lambda_1)\in\Lambda_2}\ 
\frac{1}{2\overline{\alpha}}\,x_0'\,\Pi_0(\lambda_0,\lambda_1)\,x_0
+\frac{1}{2\overline{\alpha}}\big(\alpha_0\lambda_0+\alpha_1\lambda_1\big) \label{eq:2stage-opt}
\end{equation}
where, for $\lambda_1\in\bbR$
\begin{gather*}
M_1(\lambda_1)\eqbyd
\begin{bmatrix}
B'P_fB+R & B'P_fG\\
(B'P_fG)' & G'P_fG-\lambda_1 I
\end{bmatrix}\\
\Pi_1(\lambda_1)\eqbyd
Q+A'P_fA
- A'P_f\!\begin{bmatrix}B & G\end{bmatrix}
M_1(\lambda_1)^{-1}
\begin{bmatrix}B'\\ G'\end{bmatrix}\!P_f A
\end{gather*}
Given $\lambda_1$, define
\begin{gather*}
M_0(\lambda_0,\lambda_1)\eqbyd
\begin{bmatrix}
B'\Pi_1(\lambda_1)B+R & B'\Pi_1(\lambda_1)G\\
(B'\Pi_1(\lambda_1)G)' & G'\Pi_1(\lambda_1)G-\lambda_0 I
\end{bmatrix}
\end{gather*}
\begin{align*}
\Pi_0(\lambda_0,\lambda_1) &\eqbyd Q+A'\Pi_1(\lambda_1)A\\
&\quad - A'\Pi_1(\lambda_1)\!\begin{bmatrix}B & G\end{bmatrix}
M_0(\lambda_0,\lambda_1)^{-1}\begin{bmatrix}B'\\ G'\end{bmatrix}\!\Pi_1(\lambda_1) A
\end{align*}
Given the solution to the convex optimization \eqref{eq:2stage-opt}, $(\lambda_0^*,\lambda_1^*)$, and terminal condition $P_f \succeq0$, then
\begin{enumerate}
\item The optimal control policies $u^*_0(x_0)$ and $u^*_1(x_1)$ from \eqref{opt-control-policy} satisfy the stationary conditions
\begin{align}
M_0(\lambda_0^*,\lambda_1^*)
\begin{bmatrix} u_0 \\ z_0 \end{bmatrix}^* &=- \begin{bmatrix} B' \\ G'\end{bmatrix} \Pi_1(\lambda_1^*)A \; x_0 \label{2stage-u0} \\
M_1(\lambda_1^*)
\begin{bmatrix} u_1 \\ z_1 \end{bmatrix}^* &= -\begin{bmatrix} B' \\ G'\end{bmatrix} P_fA \; x_1 \label{2stage-u1}
\end{align}
where $z_k$ denotes the Lagrangian stationary disturbance from the unconstrained stationary conditions, computed by the second block of \eqref{2stage-u0}--\eqref{2stage-u1}; the constrained optimal disturbance $w_k^*(x_k)$ is given separately in item~2. By \cref{prop:range-inclusion-invertible} and \cref{asst2,asst3}, $M_0$ and $M_1$ are invertible, so \eqref{2stage-u0}--\eqref{2stage-u1} determine $(u_k^*, z_k^*)$ uniquely.
\item The optimal disturbance policies $w^*_0(x_0) = \overline{w}_0(u^*_0) \cap \bbW_0$ and $w^*_1(x_1) = \overline{w}_1(u^*_1) \cap \bbW_1$ from \eqref{opt-dist-policy} satisfy
\begin{align}
\begin{split}
(B'\Pi_1(\lambda_1^*)G)'u^*_0(x_0) &+(G'\Pi_1(\lambda_1^*)G - \lambda_0^* I) \ \overline{w}_0(u^*_0) \\
&= - G'\Pi_1(\lambda_1^*)Ax_0
\end{split}
\label{2stage-w0} \\
\begin{split}
(B'P_fG)'u^*_1(x_1) &+( G'P_fG - \lambda_1^* I ) \ \overline{w}_1(u^*_1) \\
&= - G'P_fAx_1
\end{split}
\label{2stage-w1}
\end{align}
where $\overline{w}_0(u^*_0)$ and $\overline{w}_1(u^*_1)$ are the sets of solutions to \eqref{2stage-w0} and \eqref{2stage-w1} respectively, and $\bbW_k = \{w_k : |w_k|^2 = \alpha_k\}$.
\item The optimal cost to \eqref{2stagedp-w} is
\begin{equation}
V^*(x_0) = \frac{1}{2\overline{\alpha}}\,x_0'\,\Pi_0(\lambda_0^*,\lambda_1^*)\,x_0
+\frac{1}{2\overline{\alpha}}\big(\alpha_0\lambda_0^*+\alpha_1\lambda_1^*\big) \label{2stage-cost}
\end{equation}
\item For all $(\lambda_0,\lambda_1)\in\Lambda_2$, we have that $\Pi_0(\lambda_0,\lambda_1)\succeq0$ and $\Pi_1(\lambda_1)\succeq0$.
\end{enumerate}
\end{proposition}

\noindent\textbf{Proof sketch.}\ The proof has four blocks: replacement of the inequality constraints $|w_k|^2 \leq \alpha_k$ by equalities under \cref{asst3}, deferment of the stagewise Lagrange multipliers $\lambda_0, \lambda_1$ to the outer optimization via the stacked Lagrangian and strong duality, construction of the feasibility domain $\Lambda_2$ from the admissibility conditions $\lambda_1 \geq \norm{G'P_fG}$ and $\lambda_0 \geq \norm{G'\Pi_1(\lambda_1)G}$, and joint convexity of the value function in $(\lambda_0, \lambda_1, x_0)$. The full proof is in the appendix (\cref{app:proofs}).

\subsection{Finite Horizon Solution}
We now generalize to derive the recursive optimal solution to the finite horizon StDAR \eqref{stagedp-w}
\begin{equation*}
    V^*(x_0) \eqbyd \min_{u_0}\max_{w_0} \; \min_{u_1}\max_{w_1} \; \cdots \min_{u_{N-1}}\max_{w_{N-1}} \;
 \frac{V(x_0, \useq, \wseq)}{\overline{\alpha}} \; \; \wseq \in \bbW_{st}
\end{equation*}
where $\useq \eqbyd (u_0,u_1,\dots,u_{N-1})$, $\wseq \eqbyd (w_0,w_1,\dots,w_{N-1})$, $\overline{\alpha} = \sum_{k=0}^{N-1}\alpha_k$, and the objective function is \eqref{maincost}
\begin{equation*}
    V(x_0, \useq, \wseq) = \sum_{k=0}^{N-1} \ell(x_k, u_k)  + \ell_f(x_N)
\end{equation*}
The solution employs backward dynamic programming to evaluate the Bellman recursion \eqref{bellman-recursion} at each stage by introducing Lagrange multipliers for the constrained maximizations.

\Cref{prop:ndstage} extends \cref{prop:2stage-StDAR} to horizon $N$: the optimal multiplier vector $\boldsymbol{\lambda}_0^*(x_0)$ minimizes a convex value function over a feasibility domain $\Lambda_N$ that the proposition also constructs, and the optimal control and disturbance policies follow from the stagewise stationary conditions evaluated at $\boldsymbol{\lambda}_0^*$. The proposition characterizes the solution for a given initial state $x_0$. The state feedback policy is obtained by resolving the convex optimization \eqref{lst} at each stage for the current state, as detailed in \cref{sec:policy-structure}.

\begin{proposition}[Finite horizon StDAR \eqref{stagedp-w}]
\label{prop:ndstage}
Let Assumptions~1--3 hold; the matrices $M_k(\boldsymbol{\lambda}_k)$ and $\Pi_k(\boldsymbol{\lambda}_k)$ used below are defined later in this proposition by the recursion \eqref{stagerec1} with terminal condition $\Pi_N = P_f$. Define
\begin{align*}
\Lambda_N \eqbyd \Big\{\boldsymbol{\lambda}_0\in\bbR^N: &\ \lambda_{N-1}\ge\norm{G'P_fG},\\
&\ \lambda_k\ge\norm{G'\Pi_{k+1}(\boldsymbol{\lambda}_{k+1})G},\\
&\quad k=0,\ldots,N-2\Big\}
\end{align*}
\[
 \sum_{k=0}^{N-1} \alpha_k=\overline{\alpha}
\]
where $\boldsymbol{\lambda}_k = (\lambda_k,\ldots,\lambda_{N-1})$ denotes the vector of multipliers from stage $k$ onward. Consider the convex optimization
\begin{equation}
\mathbf{L}_{st}: \quad \min_{\boldsymbol{\lambda}_0 \in \Lambda_N}
\frac{1}{2\overline{\alpha}}\,x_0'\,\Pi_0(\boldsymbol{\lambda}_0)\,x_0
+\frac{1}{2\overline{\alpha}}\sum_{k=0}^{N-1}\alpha_k\lambda_k \label{lst}
\end{equation}
where $\Pi_k(\boldsymbol{\lambda}_k)$ is given by the recursion
\begin{equation}
\begin{split}
\Pi_k(\boldsymbol{\lambda}_k) &= Q+A'\Pi_{k+1}A-A' \Pi_{k+1} \begin{bmatrix} B & G \end{bmatrix}\\
&\quad M_k(\boldsymbol{\lambda}_k)^{-1}
\begin{bmatrix} B' \\ G'\end{bmatrix}\Pi_{k+1}A
\end{split}
\label{stagerec1}
\end{equation}
for $k \in \{0,1,\ldots,N-1\}$ with
\begin{equation}
M_k(\boldsymbol{\lambda}_k) \eqbyd \begin{bmatrix}
B'\Pi_{k+1}B + R & B'\Pi_{k+1} G \\
(B'\Pi_{k+1}G)' & G'\Pi_{k+1}G - \lambda_k I
\end{bmatrix} \label{Mk-def}
\end{equation}
and terminal condition $\Pi_N = P_f \succeq0$.

Given the solution to convex optimization \eqref{lst}, $\boldsymbol{\lambda}_0^* = (\lambda_0^*,\ldots,\lambda_{N-1}^*)$, then
\begin{enumerate}
\item The optimal control policy $u^*_k(x_k)$ from \eqref{opt-control-policy} for $k \in \{0,\ldots,N-1\}$ to \eqref{stagedp-w} satisfies the stationary conditions
\begin{equation}
M_k(\boldsymbol{\lambda}_k^*)
\begin{bmatrix} u_k \\ z_k \end{bmatrix}^* = -\begin{bmatrix} B' \\ G'\end{bmatrix} \Pi_{k+1}(\boldsymbol{\lambda}_{k+1}^*)A \; x_k \label{Nstage-uk}
\end{equation}
where $z_k$ denotes the Lagrangian stationary disturbance from the unconstrained stationary conditions, computed by the second block of \eqref{Nstage-uk}. The constrained optimal disturbance $w_k^*(x_k)$ is given separately in item~2. By \cref{prop:range-inclusion-invertible} and \cref{asst2,asst3}, $M_k(\boldsymbol{\lambda}_k^*)$ is invertible, so \eqref{Nstage-uk} determines $(u_k^*, z_k^*)$ uniquely.
\item The optimal disturbance policy $w^*_k(x_k) = \overline{w}_k(u^*_k) \cap \bbW_k$ from \eqref{opt-dist-policy} for $k \in \{0,\ldots,N-1\}$ to \eqref{stagedp-w} satisfies
\begin{equation}
\begin{split}
(B'\Pi_{k+1}(\boldsymbol{\lambda}_{k+1}^*)G)'u^*_k(x_k) &+(G'\Pi_{k+1}(\boldsymbol{\lambda}_{k+1}^*)G - \lambda_k^* I) \ \overline{w}_k(u^*_k) \\
&= - G'\Pi_{k+1}(\boldsymbol{\lambda}_{k+1}^*)Ax_k
\end{split}
\label{Nstage-wk}
\end{equation}
where $\overline{w}_k(u^*_k) \subset \bbR^q$ is the set of solutions to \eqref{Nstage-wk}, and $\bbW_k = \{w_k : |w_k|^2 = \alpha_k\}$.
\item The optimal cost to \eqref{stagedp-w} is
\begin{equation}
V^*(x_0) = \frac{1}{2\overline{\alpha}}\,x_0'\,\Pi_0(\boldsymbol{\lambda}_0^*)\,x_0
+\frac{1}{2\overline{\alpha}}\sum_{k=0}^{N-1}\alpha_k\lambda_k^* \label{ndocstage}
\end{equation}
\item For all $\boldsymbol{\lambda}_0\in\Lambda_N$, we have that $\Pi_k(\boldsymbol{\lambda}_k)\succeq0$ for $k \in \{0,\ldots,N\}$.
\end{enumerate}
\end{proposition}

\noindent\textbf{Proof sketch.}\ The proof extends \cref{prop:2stage-StDAR} to horizon $N$ by backward induction. Each stage applies the same four-block argument (constraint replacement, deferment of $\lambda_k$ via stacked Lagrangian and strong duality, convexity of the recursive feasibility domain $\Lambda_{N-k}$ from joint convexity of $\phi_{k+1}$, and joint convexity of the value function), with the induction hypothesis carrying $\Lambda_{N-k-1}$ convex, $\phi_{k+1}$ jointly convex, and $\Pi_{k+1} \succeq 0$. The full proof is in the appendix (\cref{app:proofs}).

\subsection{Optimal Policy and Implementation}
\label{sec:policy-structure}

The optimal state feedback policy for the StDAR \eqref{stagedp-w} is nonlinear in the state. To understand this nonlinearity, we first recall the Bellman recursion from \eqref{bellman-recursion}--\eqref{opt-control-policy}
\begin{equation}
V_k(x) = \min_u \max_{|w|^2 = \alpha_k} \left[ \frac{1}{\overline{\alpha}}\ell(x,u) + V_{k+1}(Ax+Bu+Gw) \right] \tag{\ref{bellman-recursion}}
\end{equation}
\begin{equation}
u_k^*(x) = \arg\min_u \max_{|w|^2 = \alpha_k} \left[ \frac{1}{\overline{\alpha}}\ell(x,u) + V_{k+1}(Ax+Bu+Gw) \right] \tag{\ref{opt-control-policy}}
\end{equation}
The dynamic programming solution in \cref{prop:ndstage} evaluates these recursions by introducing Lagrange multipliers $\boldsymbol{\lambda}_k = (\lambda_k,\ldots,\lambda_{N-1})$ for the constrained maximizations, transforming the problem into a backward recursion for the matrices $\Pi_k(\boldsymbol{\lambda}_k)$ via \eqref{stagerec1} and a forward optimization for the multipliers at each stage.

At stage $k$ with current state $x_k$, the optimal multipliers for the remaining $N-k$ stages are determined by
\begin{equation}
\boldsymbol{\lambda}_k^*(x_k; \boldsymbol{\alpha}_{k}) = \arg\min_{\boldsymbol{\lambda}_k \in \Lambda_{N-k}} \frac{1}{2\overline{\alpha}} x_k' \Pi_k(\boldsymbol{\lambda}_k) x_k + \frac{1}{2\overline{\alpha}} \sum_{j=k}^{N-1} \alpha_j \lambda_j
\label{eq:stage-k-opt}
\end{equation}
where $\overline{\alpha} = \sum_{j=0}^{N-1} \alpha_j$, $\boldsymbol{\alpha}_{k} = (\alpha_k,\ldots,\alpha_{N-1})$ denotes the remaining stage bound, and $\Pi_k(\boldsymbol{\lambda}_k)$ is computed via the backward recursion \eqref{stagerec1}. Given $\boldsymbol{\lambda}_k^*(x_k; \boldsymbol{\alpha}_{k})$, the optimal control from \eqref{Nstage-uk} is
\begin{equation}
u_k^*(x_k; \boldsymbol{\alpha}_{k}) = K_k(\boldsymbol{\lambda}_k^*(x_k; \boldsymbol{\alpha}_{k})) x_k
\label{eq:nonlinear-policy}
\end{equation}
where the gain matrix is defined by
\begin{equation}
K_k(\boldsymbol{\lambda}_k) \eqbyd -\begin{bmatrix} I & 0 \end{bmatrix} M_k(\boldsymbol{\lambda}_k)^{-1} \begin{bmatrix} B' \\ G' \end{bmatrix} \Pi_{k+1}(\boldsymbol{\lambda}_{k+1}) A
\label{eq:gain-def}
\end{equation}
The policy \eqref{eq:nonlinear-policy} is nonlinear in $x_k$ because the optimal multipliers $\boldsymbol{\lambda}_k^*(x_k; \boldsymbol{\alpha}_{k})$ depend on the state through the quadratic term in \eqref{eq:stage-k-opt}, making the composition $x_k \mapsto \boldsymbol{\lambda}_k^*(x_k; \boldsymbol{\alpha}_{k}) \mapsto K_k(\boldsymbol{\lambda}_k^*(x_k; \boldsymbol{\alpha}_{k}))$ state dependent and nonlinear.

\begin{remark}[Parametric dependence on stage bound]
The stage bound $$\boldsymbol{\alpha}_0 = (\alpha_0,\ldots,\alpha_{N-1})$$ are fixed problem data, not time-varying states. The semicolon notation in \eqref{eq:stage-k-opt}--\eqref{eq:nonlinear-policy} emphasizes that $\boldsymbol{\alpha}_{k}$ enters as a parameter: at runtime, only the current state $x_k$ is measured and used to compute the control. The mapping $u_k^*(\cdot; \boldsymbol{\alpha}_{k})$ is generally nonlinear in $\boldsymbol{\alpha}$ across problem instances: changing the stage bound yields a different policy through the optimal multipliers, gains, and propagation matrices. Within a fixed instance with given $\boldsymbol{\alpha}$, we suppress the parametric dependence and write $u_k^*(x_k)$ for brevity.
\end{remark}

\begin{remark}[Comparison with LQR]
The StDAR policy differs fundamentally from the LQR policy. The backward sweep computes $\Pi_k(\boldsymbol{\lambda}_k)$ for all $\boldsymbol{\lambda}_k \in \Lambda_{N-k}$ via \eqref{stagerec1}, but unlike LQR, these do not directly yield fixed gain matrices $K_k$ applicable as $u^*_k(x_k) = K_k x_k$. Instead, an online forward optimization solving \eqref{eq:stage-k-opt} at each stage $k$ determines state dependent gains from the current state $x_k$. The nonlinearity arises from this online optimization.
\end{remark}

\begin{remark}[Implementation and time consistency]
Implementation of the optimal policy \eqref{eq:nonlinear-policy} demands resolving \eqref{eq:stage-k-opt} at every stage $k$ using the current state $x_k$. Deviations of the realized state from the nominal trajectory, arising from disturbances, model mismatch, or other sources, invalidate the multipliers $\boldsymbol{\lambda}_0^*(x_0; \boldsymbol{\alpha}_{0})$ computed initially at $k=0$, necessitating repeated optimization. This shrinking horizon distinguishes the StDAR from problems admitting precomputed offline policies.
\end{remark}

\begin{remark}[Computational implementation]
Two implementation strategies exist for the optimal policy \eqref{eq:nonlinear-policy}:
\begin{enumerate}
\item \textbf{Online optimization:} Given the current state $x_k$ at stage $k$, solve the convex optimization \eqref{eq:stage-k-opt} to obtain $\boldsymbol{\lambda}_k^*(x_k; \boldsymbol{\alpha}_{k})$, compute the gain $K_k(\boldsymbol{\lambda}_k^*(x_k; \boldsymbol{\alpha}_{k}))$ from \eqref{eq:gain-def}, and apply $u_k = K_k(\boldsymbol{\lambda}_k^*(x_k; \boldsymbol{\alpha}_{k})) x_k$. This convex program has dimension $N-k$, which decreases as the horizon shrinks.
\item \textbf{Offline precomputation:} Discretize the state space, precompute and store the mapping $x_k \mapsto \boldsymbol{\lambda}_k^*(x_k; \boldsymbol{\alpha}_{k})$, then implement $u_k = K_k(\boldsymbol{\lambda}_k^*(x_k; \boldsymbol{\alpha}_{k})) x_k$ through table lookup at runtime. The curse of dimensionality restricts this approach to low-dimensional systems.
\end{enumerate}
Algorithm~\ref{alg:online} summarizes the online approach.
\end{remark}

\begin{algorithm}[h]
\caption{Online implementation of nonlinear optimal policy}
\label{alg:online}
\begin{algorithmic}[1]
\STATE \textbf{Input:} Horizon $N$, system matrices $(A,B,G)$, weights $(Q,R,P_f)$, bound $\boldsymbol{\alpha}_k = (\alpha_k,\ldots,\alpha_{N-1})$
\FOR{$k = 0, 1, \ldots, N-1$}
    \STATE Observe current state $x_k$
    \STATE Solve optimization \eqref{eq:stage-k-opt} to obtain $\boldsymbol{\lambda}_k^*(x_k; \boldsymbol{\alpha}_{k}) = (\lambda_k^*, \ldots, \lambda_{N-1}^*)$
    \STATE Compute gain $K_k(\boldsymbol{\lambda}_k^*(x_k; \boldsymbol{\alpha}_{k}))$ from \eqref{eq:gain-def}
    \STATE Apply control $u_k = K_k(\boldsymbol{\lambda}_k^*(x_k; \boldsymbol{\alpha}_{k})) x_k$
    \STATE System evolves: $x_{k+1} = Ax_k + Bu_k + Gw_k$
\ENDFOR
\end{algorithmic}
\end{algorithm}

\section{Steady-state StDAR}
\label{sec:ss-stage}

This section introduces the steady-state StDAR, defined directly as a fixed-point optimization for constant disturbance bound, and derives the LMI representation via \cref{prop:lmigen-stage}.

\subsection{Steady-state problem}

Let Assumptions 1-4 hold and consider constant stage bound $\alpha_k = \alpha$ for all $k$. Given the linear system \eqref{system}, we define the following optimization, denoted as steady‑state StDAR
\begin{subequations}\label{eq:ss-problem-stage}
\begin{align}
  \min_{\lambda,\Pi}\;&\; V(\lambda,\Pi)
  \label{eq:ss-problem-stage:cost}\\
  \text{s.\,t. }& \;\lambda\geq\norm{G'\Pi G}\quad
                  g(\lambda,\Pi)=0
  \label{eq:ss-problem-stage:constraints}
\end{align}
\end{subequations}
where
\[
V(\lambda,\Pi) \eqbyd \lambda /2 
\]
\begin{align*}
g(\lambda,\Pi) &\eqbyd Q+A'\Pi A -A' \Pi \begin{bmatrix} B & G \end{bmatrix} \\
&\quad \begin{bmatrix} B'\Pi B + R & B'\Pi G \\ (B'\Pi G)' & G'\Pi G - \lambda I  \end{bmatrix}^{-1}
    \begin{bmatrix} B' \\ G'\end{bmatrix}\Pi A - \Pi
\end{align*}
\[
  M(\lambda,\Pi)\eqbyd
  \begin{bmatrix}
      B' \Pi B + R & B' \Pi G\\
      G' \Pi B     & G' \Pi G - \lambda I
  \end{bmatrix}
\]

\begin{remark}[Infinite horizon average cost]
For constant stage bound $\alpha_k = \alpha$, the average cost per stage of the finite horizon problem in \cref{prop:ndstage} is
\[
\frac{1}{2N\alpha}x_0'\Pi_0(\lambda_0,\ldots,\lambda_{N-1})x_0 + \frac{1}{2N}\sum_{k=0}^{N-1}\lambda_k
\]
The steady-state problem \eqref{eq:ss-problem-stage} is the formal $N\to\infty$ limit of this average under three hypotheses: (i) $\Pi_0$ is bounded above in $N$, so the boundary term vanishes and the limit is independent of $x_0$; (ii) the time average $(1/N)\sum_{k=0}^{N-1}\lambda_k$ converges to $\overline{\lambda}$, consistent with the turnpike behavior in \cref{rem:turnpike-stage}; (iii) the Riccati recursion is continuous in $(\lambda,\Pi)$, so $\overline{\lambda}$ and the corresponding Riccati fixed point $\overline{\Pi}$ satisfy $g(\overline{\lambda},\overline{\Pi})=0$. A rigorous proof of (i)--(iii) for the StDAR remains open; for related convergence results for the SiDAR see \citep{mannini:rawlings:2026b}.
\end{remark}

Let $(\overline{\lambda},\overline{\Pi})$ denote the solution to \eqref{eq:ss-problem-stage}. 
\cref{prop:lmigen-stage} establishes an LMI for the steady-state problem \eqref{eq:ss-problem-stage}. Existence of a solution to the LMI is addressed in \cref{prop:lmi-existence-stage}.

\begin{proposition}[LMI for steady-state StDAR]
    \label{prop:lmigen-stage}
The solution to steady‑state StDAR \eqref{eq:ss-problem-stage} is implied by the following optimization
\begin{equation}
        \min_{\lambda,P,F} \; \lambda /2 \label{inflmi-stage}
    \end{equation}
    subject to
    \[
    \begin{bmatrix}
       P & (AP - BF)' & 0 & P\hat{Q}' - F'\hat{R}'\\
       AP - BF & P & G & 0 \\
       0 & G' & \lambda I & 0\\
       (P\hat{Q}' - F'\hat{R}')'  & 0 & 0 & I
    \end{bmatrix} \succeq 0
    \]    
    where $K = -F P^{-1}$, $P = \Pi^{-1}$, $\hat{Q}' = \begin{bmatrix} Q^{1/2} & 0 \end{bmatrix}$, and $\hat{R}' = \begin{bmatrix} 0 & R^{1/2} \end{bmatrix}$.
\end{proposition}

\begin{proof}
The proof follows the approach of Proposition 18 in \citep{mannini:rawlings:2026b}. The stage bound problem in this paper requires only the single $4\times 4$ LMI constraint above, whereas the signal bound problem in Proposition 18 requires both the $4\times 4$ LMI and an additional $2\times 2$ LMI constraint. With this modification, the remainder of the proof carries through unchanged.
\end{proof}

\begin{proposition}[Existence of LMI]
\label{prop:lmi-existence-stage}
Let Assumptions~\ref{asst1}--\ref{asst4} hold. The optimization problem \eqref{inflmi-stage} has an optimal solution $(\lambda^*, P^*, F^*)$.
\end{proposition}

\begin{proof}
The proof follows the approach of Proposition 19 in \citep{mannini:rawlings:2026b}. The stage bound problem in this paper requires only the single $4\times 4$ LMI constraint, whereas the signal bound problem in Proposition 19 considers both the $4\times 4$ LMI and an additional $2\times 2$ LMI constraint. With this modification, the compactness and existence arguments carry through unchanged.
\end{proof}

\section{Numerical Examples}
\label{sec:num-example}

The following examples illustrate the theoretical properties of the StDAR and compare its solution to the signal bound disturbance attenuation regulator (SiDAR) developed in \citep{mannini:rawlings:2026b}. While this paper focuses on stage bound constraints, comparing the two problems reveals fundamental differences in their optimal policies and computational requirements. While all analytical results apply to arbitrary dimension $n$, scalar cases are presented to facilitate visualization.

The signal bound problem from \citep{mannini:rawlings:2026b} imposes a single aggregate constraint over the entire horizon
\[
	\bbW_{si}
	\eqbyd
	\Bigl\{
	\wseq
	 \mid
	\sum_{k=0}^{N-1}|w_k|^2 \leq \alpha
	\Bigr\}
\]
and poses the disturbance attenuation optimization
\begin{equation}
    V^*(x_0) \eqbyd \min_{u_0}\max_{w_0} \; \min_{u_1}\max_{w_1} \; \cdots \min_{u_{N-1}}\max_{w_{N-1}} \;
 \frac{V(x_0, \useq, \wseq)}{\alpha} \; \; \wseq \in \bbW_{si} \label{signaldp}
\end{equation}
subject to \eqref{system}. As for the stage bound problem \eqref{stagedp-w}, under \cref{asst3} the maximum over the disturbance is attained on the boundary $\sum_{k=0}^{N-1}|w_k|^2 = \alpha$, so the normalized cost equals the attenuation ratio $V(x_0,\useq,\wseq)/\sum_{k=0}^{N-1}|w_k|^2$ at the worst case and the optimization is well-defined. In contrast to the stage bound problem \eqref{stagedp-w}, the signal bound problem uses a single scalar multiplier $\lambda$ for the entire horizon, while the stage bound problem employs a vector of multipliers $\boldsymbol{\lambda}_0 = (\lambda_0,\ldots,\lambda_{N-1})$ with one per stage.

\subsection{StDAR vs SiDAR Comparison}
Consider the scalar system
\begin{equation*}
   \begin{aligned}
   	A = 1 \; \; \; B = 1 \; \; \; G = 1 \; \; \; R = 1 \; \; \; Q = 0.2 \; \; \; P_f=0.25
   \end{aligned} 
\end{equation*}
and disturbance bound
\begin{align*}
&|w_k|^2 = 1 \;\text{ for all } k \quad\text{(StDAR \eqref{stagedp-w})}\\
&\sum_{k=0}^{N-1} |w_k|^2 = N \quad\text{(SiDAR \eqref{signaldp})}
\end{align*}
chosen so that both problems have equivalent total disturbance energy budgets. For this system the steady-state LQR Riccati matrix is $P_{LQR}=0.55$, and the steady-state $H_{\infty}$ state feedback Riccati matrix is $P_{\infty}=1.2$.
\begin{figure*}[t]
\centering
\includegraphics[width=\textwidth]{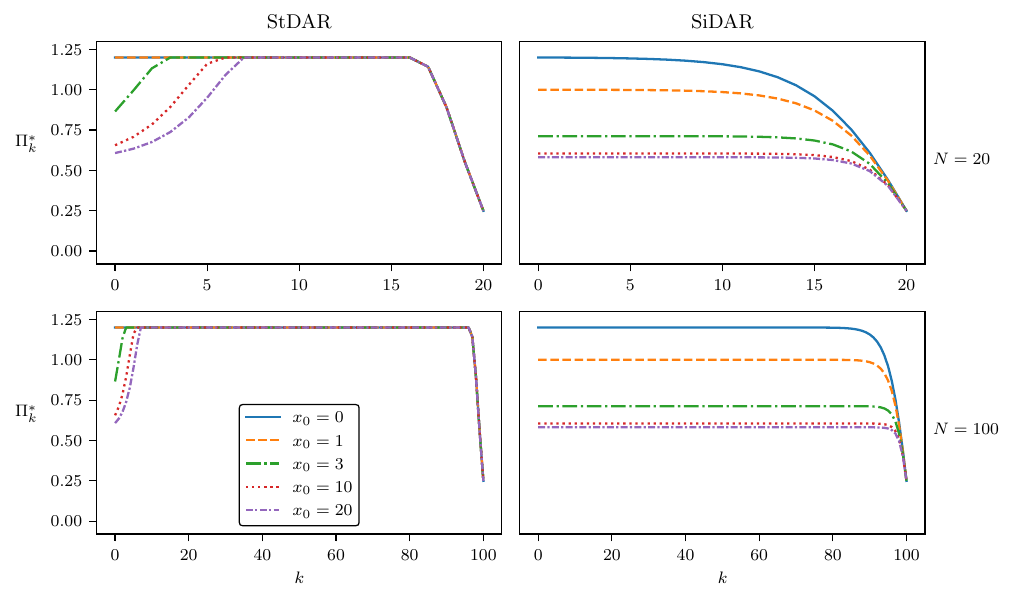}
\caption{$\Pi_k$ as a function of stage $k$ for the finite horizon StDAR \eqref{stagedp-w} (left) and SiDAR \eqref{signaldp} (right), for $N=20$ (top) and $N=100$ (bottom). The dashed black line represents the $H_{\infty}$ steady-state value $P_{\infty}=1.2$; the solid black line represents the LQR value $P_{LQR}=0.55$.}
\label{fig:1dplot}
\end{figure*}
We solve both problems for initial states $x_0 \in \{0, 1, 3, 10, 20\}$ with horizon lengths $N=20$ and $N=100$. For the SiDAR, state $\{0\}$ lies in the linear control region $\mathcal{X}_L$ and states $\{1, 3, 10, 20\}$ lie in the nonlinear control region $\mathcal{X}_{NL}$.

\Cref{fig:1dplot} shows the optimal $\Pi_k$ as a function of stage $k$ for both problems. For the SiDAR (right panels), initial states in $\mathcal{X}_L$ converge quickly to $P_{\infty} = 1.2$, while states in $\mathcal{X}_{NL}$ exhibit monotonic convergence toward values between $P_{\infty}$ and $P_{LQR}$ depending on $x_0$. The recursion is monotonic in $k$ for all trajectories.

In contrast, the StDAR (left panels) exhibits fundamentally different behavior. For larger initial states, the recursion displays a nonmonotonic turnpike: $\Pi_k$ initially varies between $P_{\infty}$ and $P_{LQR}$, then stabilizes near $P_{\infty}$ for the majority of the horizon, before departing to satisfy the terminal condition $P_f=0.25$. This turnpike behavior becomes more pronounced with longer horizons, with the plateau length increasing as $N$ grows. For smaller initial states, the recursion converges more directly to $P_{\infty}$ and remains there for most of the horizon.

This fundamental difference stems from the disturbance constraint. The stage bound problem allows distinct multipliers $\lambda_k$ at each stage, enabling different recursion patterns depending on the initial state magnitude. In contrast, the signal bound problem enforces a single multiplier $\lambda$ across the entire horizon, yielding monotonic convergence for all trajectories.

\begin{remark}[Non-monotonic turnpike]\label{rem:turnpike-stage}
When the stage bound sequence is time invariant ($\alpha_k \eqbyd \alpha$ for all $k$), numerical evidence across all tested instances indicates that the StDAR backward recursion $\Pi_{k}$ exhibits a nonmonotonic turnpike for sufficiently large initial states: independently of $x_0$, the recursion rapidly approaches the steady-state $H_{\infty}$ value, remains close to it for most of the horizon, and departs only near the endpoints to satisfy boundary conditions. The plateau length increases with $N$. Establishing general conditions under which this pattern holds remains an open question.
\end{remark}

\subsection{Stage bound vs Signal bound Policy Comparison}
Consider the scalar system
\begin{equation*}
A=0.5 \; \; \; B=1 \; \; \; G=1 \; \; \; R=1 \; \; \; Q=0.25 \; \; \; P_f=0.25
\end{equation*}
with horizon $N=20$. \Cref{fig:stage_vs_signal} compares the optimal control policies for the StDAR \eqref{stagedp-w} with stage bound $\alpha_k = 0.05$ for all $k$ (top panel) and the SiDAR \eqref{signaldp} with fixed disturbance budget $b_0=1$ (bottom panel). Note that $b_0=\alpha$. The stage bound policy $u^*_0(x_0; \boldsymbol{\alpha}_{0:N-1})$ depends on the vector of remaining stage bound as a parameter, while the signal bound policy $u^*_0(x_0, b_0)$ depends on the budget state $b_0$ as a measured state variable.

The bottom panel shows the SiDAR linear region $\mathcal{X}_L$ (shaded) where the control policy is linear in the state. Outside this region, the SiDAR policy is nonlinear. The stage bound problem (top panel) exhibits qualitatively different optimal control. This difference stems from the constraint: stage bound constraints impose separate limits at each time step via a vector of multipliers $\boldsymbol{\lambda}_0 = (\lambda_0,\ldots,\lambda_{N-1})$, while signal bound constraints impose a single aggregate limit via a scalar multiplier $\lambda$.
\begin{figure}
\centering
\includegraphics[width=1\linewidth]{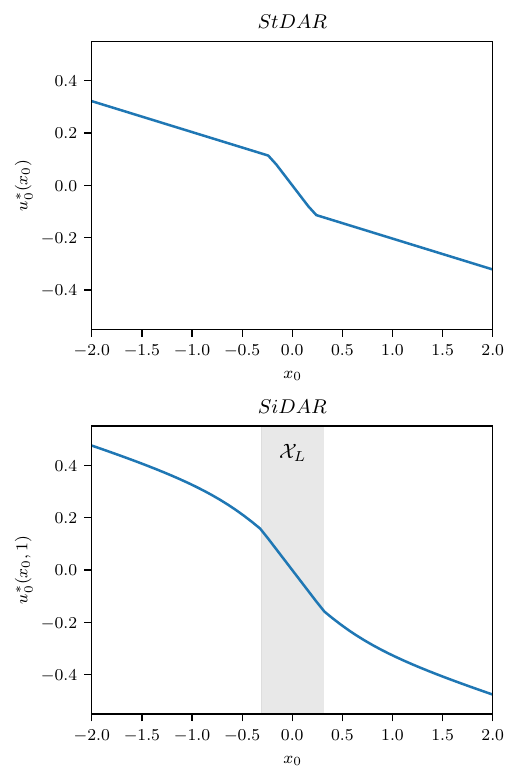}
\caption{Optimal control policies for scalar system with $N=20$. Top: stage bound $u^*_0(x_0)$ with $\alpha_k = 0.05$ for all $k$. Bottom: signal bound $u^*_0(x_0, b_0)$ with budget $b_0=\alpha=1$. Shaded region in bottom panel indicates $\mathcal{X}_L$ where SiDAR policy is linear in $x_0$.}
\label{fig:stage_vs_signal}
\end{figure}

\subsection{Computational Complexity}
To assess computational performance, the LMI optimization from \eqref{inflmi-stage} was evaluated for state and input dimensions $n = m$ ranging from 2 to 78 (in increments of 4) on randomly generated stable systems. Four independent problem instances were solved at each dimension using the MOSEK solver~\citep{mosek:2025}. The empirical runtime is approximately cubic in $n$ for the dimensions tested. The worst-case theoretical complexity of interior point methods on semidefinite programs exceeds cubic in $n$ in general \citep[Ch.~11]{boyd:vandenberghe:2004}; the observed cubic scaling provides numerical evidence that the LMI based steady-state StDAR problem is computationally tractable for the dimensions tested.

\section{Summary}
\label{sec:end}
This paper derives a finite horizon recursive feedback solution to the StDAR for linear systems via dynamic programming. The problem accommodates arbitrary initial states throughout the state space, addressing theoretical gaps in existing literature where solutions were valid only at the origin.

The optimal state feedback policy is nonlinear in the state and requires solving a tractable convex optimization for the Lagrange multiplier vector at each stage; the control gain is then explicit. Unlike the SiDAR, where a single scalar multiplier enforces the aggregate budget, the StDAR requires a vector of multipliers—one for each remaining stage.

For systems with constant stage bound, the steady-state StDAR reduces to a tractable linear matrix inequality whose empirical computational cost is approximately cubic in $n$.

Numerical examples demonstrate that the backward recursion for StDAR exhibits qualitatively different behavior from the SiDAR recursion, including nonmonotonic turnpike phenomena for sufficiently large initial states.

The StDAR differs fundamentally from the SiDAR. Signal bound constraints impose a single aggregate energy limit over the entire horizon via a scalar multiplier. Stage bound constraints impose separate limits at each time step via a vector of Lagrange multipliers. This difference produces distinct solution properties and naturally accommodates time-varying disturbance bounds without requiring temporal coordination of disturbance realizations.

\section{Appendix}
\label{sec:props}

In this appendix, we compile the fundamental results used throughout this paper. To compactly state the results in this section it is convenient to define two functions
\begin{align}
M(\lambda) &\eqbyd \begin{bmatrix} M_{11} & M_{12} \\ M'_{12} & M_{22}-\lambda I \end{bmatrix} \label{eq:Mlam}\\
L(\lambda) &\eqbyd - (1/2) d'M^{\dagger}(\lambda)d + \lambda/2 \label{eq:Llam}
\end{align}
and the definition of the Schur complements
\begin{align*}
\tM_{11}(\lambda) \eqbyd (M_{22}-\lambda I) - M_{12}'M_{11}^{\dagger} M_{12} \\
\tM_{22}(\lambda) \eqbyd M_{11}-M_{12}(M_{22}-\lambda I)^{\dagger} M_{12}'
\end{align*}
The proofs for propositions \ref{prop:conlag}--\ref{prop:conminmaxalt} are found in Rawlings et al. \citep{rawlings:mannini:kuntz:2024,rawlings:mannini:kuntz:2024b}. The remaining propositions include their proofs.

We present the following classical result, stated here without proof, to justify the interchange of minimization and maximization.

\begin{theorem}[Minimax Theorem]
\label{th:minimax}
Let $U \subset \bbR^m$ and $W \subset \bbR^q$ be compact convex sets. If 
$V: U \times W \to \bbR$ is a continuous function that is convex-concave, i.e.,
$V(\cdot ,w):U \to \bbR$ is convex for all $w \in W$, and
$V(u, \cdot ):W \to \bbR$ is concave for all $u \in U$\\
Then we have that
\begin{equation*}
\min_{u \in U} \max_{w \in W} V(u,w) = \max_{w \in W} \min_{u \in U} V (u,w) 
\end{equation*}
\end{theorem}

We establish the equivalence between constrained optimization and Lagrangian minmax optimizations.

\begin{proposition}[Constrained minimization]
\label{prop:conlag}
Let $U \subseteq \bbR^n$ be a nonempty compact set and $V(\cdot): U \rightarrow \bbR$  be a continuous function on $U$, and $L(\cdot): U \times \bbR \rightarrow \bbR$ be defined as $L(u, \lambda) \eqbyd V(u) - \lambda \rho(u, U)$, where $\rho(u, U)$ denotes the distance from point $u$ to set $U$.
Consider the constrained optimization problem
\begin{equation}
\inf_{u \in U} V(u)
\label{eq:conmin}
\end{equation}
and the (unconstrained) Lagrangian minmax problem
\begin{equation}
\infp_{u} \sup_\lambda L(u,\lambda) 
\label{eq:Lgen}
\end{equation}
\begin{enumerate}
\item Solutions to both problems exist.
\item Let $V^*$ be the solution and $u^0$ be the set of optimizers of $\min_{u\in U} V(u)$. Let $L^*$ be the solution and $u^*$ the set of optimizers of $\min_u \max_\lambda L(u, \lambda)$.  Then 
\begin{equation*}
V^* = L^* \qquad u^0=u^* \qquad \olambda(u^*) = \bbR
\end{equation*}
\end{enumerate}
\end{proposition}
We solve the constrained maximization of quadratic functions over unit spheres.

\begin{proposition}[Constrained quadratic optimization]
\label{prop:conquad}

Define the \textit{convex} quadratic function, $V(\cdot): \bbR^{n} \rightarrow \bbR$ and compact constraint set $\bbW$
\begin{equation*}
V(w) \eqbyd \half w'Dw + w'd \qquad 
\bbW \eqbyd \{ w \mid w'w = 1\}
\end{equation*}
with $D \in \bbR^{n\times n} \geq 0$. Consider the constrained maximization problem
\begin{equation}
\max_{w \in \bbW} V(w)
\label{eq:maxcon}
\end{equation}
Define the Lagrangian function
\begin{equation*}
L(w, \lambda) = V(w) - \half \lambda (w'w-1)
\end{equation*}
and the (unconstrained) Lagrangian problem
\begin{equation}
\max_{w} \min_{\lambda}  L(w, \lambda)
\label{eq:maxmincon}
\end{equation}
and the (unconstrained) dual Lagrangian problem
\begin{equation}
\min_{\lambda} \max_{w} L(w, \lambda)
\label{eq:minmaxcon}
\end{equation}

\begin{enumerate}
\item  Solutions to all three problems \eqref{eq:maxcon}, \eqref{eq:maxmincon}, and \eqref{eq:minmaxcon} exist for all $D \geq 0$ and $d \in \bbR^n$ with optimal value
\begin{equation*}
V^* = L^0 = - (1/2) d'(D-\lambda_P I)^{\dagger} d + \lambda_P/2
\end{equation*}
where
\begin{equation}
\lambda_P \eqbyd  \text{largest real eigenvalue of $P$},
\quad 
P \eqbyd \begin{bmatrix} D & I \\ dd' & D \end{bmatrix}
\label{eq:Pdef}
\end{equation}

\item Problems \eqref{eq:maxmincon} and \eqref{eq:minmaxcon} satisfy strong duality,
and the function $L(w, \lambda)$ has saddle points (sets) $(w^*, \lambda^*)$ given by
\begin{align*}
w^* &= \begin{cases} \bigg( -(D-\lambda_P I)^{\dagger} d + N(D-\lambda_P I) \bigg)  \cap W, \; \; 
 \lambda_P = \norm{D} \\
                     -(D-\lambda_P I)^{-1} d, \qquad \lambda_P > \norm{D}
\end{cases} \\
\lambda^* &= \lambda_P
\end{align*}

\item The optimizer of \eqref{eq:maxcon} is given by $w^0 = w^*$. 

\item 
The optimizer of \eqref{eq:maxmincon} is given by 
\begin{equation*}
w^0 = w^* \qquad 
\underline{\lambda}^0(w^0) = \bbR
\end{equation*}

\item  The optimizer of \eqref{eq:minmaxcon} is given by
\begin{align*}
\ow^0(\lambda^0) &= \begin{cases} -(D-\lambda_P I)^{\dagger} d + N(D-\lambda_P I), \ & \lambda_P = \norm{D} \\
                     -(D-\lambda_PI)^{-1} d, \ & \lambda_P > \norm{D}
\end{cases} \\
\lambda^0 &= \lambda_P 
\end{align*}

\item Additionally $\lambda_P = \norm{D}$ if and only if (i) $d \in \mathcal{R}( D - \norm{D} I)$ and (ii) $\norm{(D-\norm{D}I)^{\dagger} d} \leq 1$. If (i) or (ii) do not hold, then $\lambda_P > \norm{D}$ and $\norm{(D-\lambda_P I)^{-1}d} = 1$.
\end{enumerate}
\end{proposition}
We establish the existence and solutions of unconstrained minmax and maxmin problems for quadratic functions.
\begin{proposition}[Minmax of quadratic functions]
\label{prop:sdparalt}

Consider the Lagrangian function $L(\cdot): \bbR^{m+n+1} \rightarrow \bbR$
\begin{align*}
L(u,w,\lambda) &\eqbyd (1/2) \begin{bmatrix} u \\ w \end{bmatrix}'
M(\lambda) \begin{bmatrix} u \\ w \end{bmatrix}  +
 \begin{bmatrix} u \\ w \end{bmatrix}' d + \lambda/2
\end{align*}
with $M_{11}>0$, $M(0) \geq 0$, and the two problems
\begin{equation}
\min_u \max_w L(u,w,\lambda) \qquad \max_w \min_u L(u,w,\lambda)
\label{eq:towprobs}
\end{equation} 

We establish existence of solutions as a function of (decreasing) parameter $\lambda$.
\begin{enumerate}
\item For $\lambda > \norm{M_{22}}$: Solutions to both problems exist for all $d \in \bbR^{m+n}$.

\item  For $\lambda = \norm{M_{22}}$, we have the following two cases:
\begin{enumerate}
 \item For $d \in R(M(\norm{M_{22}}))$: The solutions to both problems exist.
 \item For $d \notin R(M(\norm{M_{22}}))$: Neither problem has a solution.
\end{enumerate} 

If $M(0)$ is such that $\normf{\tM_{11}} < \norm{M_{22}}$, then we have the following cases.
\item \label{case:int}
 For $\normf{\tM_{11}} < \lambda < \norm{M_{22}}$: Only the solution to the $\max_w \min_u L$ problem exists, and it exists for all $d \in \bbR^{m+n}$.
\item \label{case:edge}
 For $\lambda = \normf{\tM_{11}}$, we have the following two cases:
\begin{enumerate}
 \item For $d \in R(M(\normf{\tM_{11}}))$: Only the solution to the $\max_w \min_u L$ problem exists.
 \item For $d \notin R(M(\normf{\tM_{11}}))$: Neither problem has a solution.
\end{enumerate} 

\item For $\lambda < \normf{\tM_{11}}$:  Neither problem has a solution. 
\end{enumerate} 
If $M(0)$ is such that $\normf{\tM_{11}} = \norm{M_{22}}$, then cases \ref{case:int} and \ref{case:edge} do not arise.
    
For $d \in \mathcal{R}(M(\lambda))$ denote the stationary points $(u^*(\lambda), w^*(\lambda))$ by
\begin{equation*}
\begin{bmatrix}u^* \\ w^* \end{bmatrix}(\lambda) \in -M^{\dagger}(\lambda) d + N(M(\lambda))
\end{equation*} 
When solutions to the respective problems exist, we have that 
\begin{gather*}
u^*(\lambda) = \arg \min_u \max_w L(u,w,\lambda) \\
w^*(\lambda) = \arg \max_w \min_u L(u,w,\lambda) \\
L^0(\lambda) = -(1/2)  d' M^{\dagger}(\lambda) d + \lambda/2 
\end{gather*} 
and the inner optimizations, $\uu^0(w^*(\lambda))$ and $\ow^0(u^*(\lambda))$, are given by all solutions  to
\begin{equation*}
 \begin{array}{lccccl}
M_{11}  &\uu^0(w^*)  &+  &M_{12}             &w^*        &= - d_1\\
M_{12}'  &u^*        &+  &(M_{22}-\lambda I) &\ow^0(u^*) &=  -d_2
\end{array} 
\end{equation*} 
or, after solving,
\begin{equation*}
 \begin{array}{lcccl}
\uu^0(w^*) &= &-M_{11}^{\dagger} (M_{12}w^* + d_1) &+ &N(M_{11})\\
\ow^0(u^*) &= &-(M_{22}-\lambda I)^{\dagger} (M_{12}'u^* + d_2) &+ &N(M_{22}-\lambda I)
\end{array} 
\end{equation*} 
    
\end{proposition}

\begin{corollary}
\label{cor:sdparalt}
    An unstated implication of \cref{prop:sdparalt}, is that case 2(b) is incompatible with either case 4(a) or 4(b). 
    In other words,  if $M(0)$ is such that $\normf{\tM_{11}} < \norm{M_{22}}$, then $M(\norm{M_{22}})$ is full rank, and its range is therefore $\bbR^{m+n}$, and there is no $d$ satisfying case 2(b).  We establish that implication here. 
\end{corollary}

\begin{proof}
First  note that from item 3 in the proof of \cref{prop:sdparalt}, we know that $M(\lambda)$ is full rank if and only if $M_{11}$ is full rank and $\tM_{11}(\lambda) \eqbyd \tM_{11} - \lambda I$ is full rank. We have that $M_{11}$ is full rank since $M_{11} > 0$ by assumption.  From its definition, we know that $\tM_{11}(\lambda)< 0$ for $\lambda > \normf{\tM_{11}}$, and is therefore full rank, which gives that $M(\lambda)$ is full rank for $\lambda > \normf{\tM_{11}}$.  

Note also that since $\tM_{11} = M_{22} - M_{12}'M_{11}^+M_{12}$, we have that $\tM_{11} \leq M_{22}$ and $\normf{\tM_{11}} \leq \norm{M_{22}}$ as well.  So $M(\norm{M_{22}})$ is full rank unless $\norm{M_{22}} = \normf{\tM_{11}}$.
\end{proof}

We extend the minmax analysis to constrained quadratic optimization problems.
\begin{proposition}[Minmax of constrained quadratic functions]
\label{prop:conminmaxalt}
Consider quadratic function $V(\cdot): \bbR^{m+n} \rightarrow \bbR$, compact constraint set $\bbW \eqbyd \{ w \mid w'w = 1\}$, and 
Lagrangian function $L(\cdot): \bbR^{m+n+1} \rightarrow \bbR$
\begin{align*}
V(u,w) &= (1/2) \begin{bmatrix} u \\ w \end{bmatrix}'
M(0)
\begin{bmatrix} u \\ w \end{bmatrix}
+  \begin{bmatrix} u \\ w \end{bmatrix}' d\\
L(u,w,\lambda) &= (1/2) \begin{bmatrix} u \\ w \end{bmatrix}'
{M(\lambda)}
\begin{bmatrix} u \\ w \end{bmatrix}  +
\begin{bmatrix} u \\ w \end{bmatrix}'  d
+  \lambda/2
\end{align*} 
We consider the two constrained optimization problems
\begin{alignat}{2}
&\min_u &\max_{w \in \bbW} &V(u,w)  \quad \text{robust control} \label{eq:conminmaxalt}\\
&\max_{w \in \bbW} &\min_u &V(u,w) \quad \text{worst-case feedforward control} \label{eq:conmaxminalt}
\end{alignat}
Assume $M(0) \geq 0$ and $M_{11} > 0$.
For $d \in \mathcal{R}(M(\lambda))$ denote stationary points by $(u^*(\lambda), w^*(\lambda))$ and evaluated Lagrangian function  
\begin{align*}
\begin{bmatrix}u^*(\lambda) \\ w^*(\lambda) \end{bmatrix} &\in -M^{\dagger}(\lambda) d + N(M(\lambda)) \\
L(\lambda) &= V(u^*,w^*) -(1/2) \lambda((w^*)'w^* -1) \\
&= -(1/2) d'M^{\dagger}(\lambda)d + \lambda/2 
\end{align*}
We  then have the following results. 
\begin{enumerate}
\item
The solution to problem \eqref{eq:conminmaxalt} exists for all $d \in \bbR^{m+n}$ and is given by
\begin{equation*}
u_r^0 = u^*(\lambda_r^0) \qquad w_r^0 = \ow^0 \cap \bbW
\end{equation*} 
where $\lambda^0_r$ denotes the solution to the following optimization, which exists for all $d \in \bbR^{m+n}$
\begin{equation}
 \lambda_r^0 = \arg \min_{\lambda \geq \norm{M_{22}}} L(\lambda)
 \label{eq:lamr}
\end{equation} 
and $\ow^0$ is all solutions to 
\begin{equation*}
M'_{12} u^*(\lambda^0_r) + (M_{22}- \lambda^0_r I) \ow^0 = -d_2
\end{equation*} 
The optimal cost is given by $V(u_r^0, w_r^0) = L(\lambda_r^0)$.

\item
The solution to problem \eqref{eq:conmaxminalt} exists for all $d \in \bbR^{m+n}$ and is given by
\begin{equation*}
u_f^0 = \uu^0 \qquad w_f^0 = w^*(\lambda^0_f) \cap \bbW
\end{equation*} 
where $\lambda^0_f$ denotes the solution to the following optimization, which exists for all $d \in \bbR^{m+n}$
\begin{equation}
 \lambda_f^0 = \arg \min_{\lambda \geq \normf{\tM_{11}}} L(\lambda)
 \label{eq:lamf}
\end{equation} 
and $\uu^0$ is all solutions to 
\begin{equation*}
M_{11} \uu^0 + M_{12} w^*(\lambda^0_f) = -d_1
\end{equation*} 
The optimal cost is given by $V(u_f^0, w_f^0) = L(\lambda_f^0)$.
\end{enumerate}
\end{proposition} 
We provide a necessary and sufficient condition for block matrix nonsingularity.

\begin{proposition}[Nonsingularity condition for block matrix]
\label{prop:nonsingular}
Let $A \succ 0$ and $D \preceq 0$. The matrix
\begin{equation*}
   M = \begin{bmatrix}
   A & B \\ B' & D  \end{bmatrix}
\end{equation*}
is invertible if and only if $\mathcal{N}(B) \cap \mathcal{N}(D) = \{0\}$.
\end{proposition}

\begin{proof}
Using the partitioned matrix determinant formula, since $A \succ 0$,  the matrix $M$ is invertible if and only if its Schur complement $\tM \eqbyd D - B'A^{-1}B$ is nonsingular.
To establish sufficiency, assume $\mathcal{N}(B) \cap \mathcal{N}(D) = \{0\}$ and let $v$ satisfy $\tM v = 0$. Then
\begin{equation*}
v'\tM v = 0 = v'Dv - (Bv)'A^{-1}(Bv)
\end{equation*}
Since $D \preceq 0$ and $A^{-1} \succ 0$, we have that $v'Dv \leq 0$ and $(Bv)'A^{-1}(Bv) \geq 0$, and since they are equal to each other, both are zero: $v'Dv = (Bv)'A^{-1}(Bv) = 0$.  Since $-D$ is positive semidefinite and $A^{-1}$ is positive definite, they both have square roots, which implies that $\sqrt{-D}v = 0$ and $A^{-1/2}Bv = 0$. Multiplying by $\sqrt{-D}$ and $A^{1/2}$, respectively, then gives $-Dv=0$ and $Bv=0$. Therefore  $v \in \mathcal{N}(B) \cap \mathcal{N}(D) = \{0\}$. Hence $v = 0$, $\tM$ is nonsingular, and $M$ is nonsingular.

To establish necessity, let $\mathcal{N}(B) \cap \mathcal{N}(D) \neq \{0\}$ and pick nonzero $v$ in the intersection. Then $Bv = 0$ and $Dv = 0$, so $\tM v = Dv - B'A^{-1}Bv = 0$ for nonzero $v$, so $\tM$ and therefore $M$ are singular.
\end{proof}

\begin{proposition}[Invertibility under range inclusion]
\label{prop:range-inclusion-invertible}
Let ${\Pi}\succeq 0$, $R\succ 0$, and ${\lambda}>0$. Assume $G'{\Pi}G-{\lambda}I\preceq 0$. If $\mathcal{R}(G)\subseteq \mathcal{R}(B)$ (equivalently, $\mathcal{N}(B')\subseteq \mathcal{N}(G')$), then the block matrix
\[
   M \eqbyd
   \begin{bmatrix}
      B'{\Pi}B + R & B'{\Pi}G \\
      G' \Pi B & G'{\Pi}G - {\lambda} I
   \end{bmatrix}
\]
is nonsingular.
\end{proposition}

\begin{proof}
Write $A\eqbyd B'{\Pi}B+R\succ 0$, $C\eqbyd B'{\Pi}G$, and $D\eqbyd G'{\Pi}G-{\lambda}I\preceq 0$, so that
$M=\begin{bmatrix} A & C \\ C' & D \end{bmatrix}$. Using the partitioned matrix determinant formula, since $A\succ 0$, the matrix $M$ is invertible if and only if its Schur complement $\tM \eqbyd D - C' A^{-1} C$ is nonsingular.

To establish sufficiency, assume $\mathcal{R}(G)\subseteq \mathcal{R}(B)$ and let $v$ satisfy $\tM v = 0$. Then
\[
v'\tM v \;=\; 0 \;=\; v'Dv - (Cv)'A^{-1}(Cv).
\]
Since $D\preceq 0$ and $A^{-1}\succ 0$, we have $v'Dv\le 0$ and $(Cv)'A^{-1}(Cv)\ge 0$, and because they are equal, both are zero:
$v'Dv=(Cv)'A^{-1}(Cv)=0$. As $-D\succeq 0$ and $A^{-1}\succ 0$ admit square roots, this implies $\sqrt{-D}\,v=0$ and $A^{-1/2}Cv=0$, hence $Dv=0$ and $Cv=0$.

From $Cv=B'{\Pi}Gv=0$ we obtain ${\Pi}Gv\in\mathcal{N}(B')\subseteq\mathcal{N}(G')$, so $G'\Pi Gv=0$. Together with $Dv=(G'{\Pi}G-{\lambda}I)v=0$ and ${\lambda}>0$, it follows that $v=0$. Therefore $\tM$ is nonsingular, and hence $M$ is nonsingular.
\end{proof}

\begin{lemma}[Riccati equalities]
\label{receq}
The equality 
\begin{align*}
\Pi(\lambda) &= Q+A'\Pi A- A' \Pi \begin{bmatrix} B & G \end{bmatrix} \\
&\quad \begin{bmatrix} B'\Pi B + R & B'\Pi G \\ (B'\Pi G)' & G'\Pi G - \lambda I  \end{bmatrix}^{\dagger}
    \begin{bmatrix} B' \\ G'\end{bmatrix}\Pi A
\end{align*}
    can be rewritten as
\begin{equation}
            \Pi = \bar{Q} + \bar{A}' \Pi\bar{A} - \bar{A}'\Pi G(G'\Pi G-\lambda I)^{\dagger}G'\Pi\bar{A} \label{xlrec}
        \end{equation}
    where $\bar{A} = A+BK$ and $\bar{Q} = Q + K'RK$ and $K$ satisfies
    \begin{equation}
  \begin{bmatrix}
        B'\Pi B+R & B'\Pi G \\
        G'\Pi B & G'\Pi G-\lambda I
    \end{bmatrix}\begin{bmatrix}
K \\ J
\end{bmatrix} =-\begin{bmatrix}
B'\Pi A \\ G'\Pi A
\end{bmatrix} \label{JK}
\end{equation}
\end{lemma}

\begin{proof}
From $M^\dagger M M^\dagger=M^ \dagger$ we have
 \begin{align}
         \Pi(\lambda) &=Q+A'\Pi A- A' \Pi \begin{bmatrix} B & G \end{bmatrix}
     M(\lambda)^{\dagger}
    \begin{bmatrix} B' \\ G'\end{bmatrix}\Pi A  \label{recjk} \\ &= Q+A'\Pi A- A' \Pi \begin{bmatrix} B & G \end{bmatrix}
     M(\lambda)^{\dagger}M(\lambda)M(\lambda)^{\dagger}
    \begin{bmatrix} B' \\ G'\end{bmatrix}\Pi A \nonumber
    \end{align}
where
\[
M(\lambda) = \begin{bmatrix} B'\Pi B + R & B'\Pi G \\ (B'\Pi G)' & G'\Pi G - \lambda I  \end{bmatrix}
\]
Define
\[
\begin{bmatrix}
        B'\Pi B+R & B'\Pi G \\
        G'\Pi B & G'\Pi G-\lambda I
    \end{bmatrix}\begin{bmatrix}
K \\ J
\end{bmatrix} =-\begin{bmatrix}
B'\Pi A \\ G'\Pi A
\end{bmatrix}
\]
or equivalently, with $b = \begin{bmatrix} B'\Pi A \\ G'\Pi A \end{bmatrix}$
    \begin{equation}
        \begin{bmatrix}
K \\ J
\end{bmatrix} =-M(\lambda)^{\dagger}b+\mathcal{N}(M(\lambda)) \label{JK2}
    \end{equation}
For any $v \in \mathcal{N}(M(\lambda))$, we have $M(\lambda)v = 0$, which gives $v'M(\lambda)v = 0$ and $(M(\lambda)^{\dagger}b)'M(\lambda)v = b'M(\lambda)^{\dagger}M(\lambda)v = 0$. Therefore, when substituting $\begin{bmatrix} K \\ J \end{bmatrix} = -M(\lambda)^{\dagger}b + v$ into the quadratic form $\begin{bmatrix} K' & J' \end{bmatrix} M(\lambda) \begin{bmatrix} K \\ J\end{bmatrix}$, all terms involving $v$ vanish. Thus, the following expression
\begin{equation}
        \Pi(\lambda) =Q+A'\Pi A- \begin{bmatrix} K' & J' \end{bmatrix}
     M(\lambda)
    \begin{bmatrix} K \\ J\end{bmatrix} \label{kj2}
\end{equation}
is equivalent to \eqref{recjk}. Expanding \eqref{kj2}
\begin{align*}
    \Pi(\lambda) = & Q+A'\Pi A-K'B'\Pi BK -K'RK-K'B'\Pi GJ-J'G'\Pi BK\\&-J'(G'\Pi G-\lambda I)J
\end{align*}
Consider
\[
B'\Pi GJ = -B' \Pi A - (B'\Pi B + R)K
\]
and
\[
J = -(G' \Pi G-\lambda I)^{\dagger}G' \Pi (A+BK)+\mathcal{N}(G' \Pi G-\lambda I)
\]
For any $q \in \mathcal{N}(G' \Pi G-\lambda I)$, we have $(G' \Pi G-\lambda I)q = 0$, giving $q'(G' \Pi G-\lambda I)q = 0$ and $((G' \Pi G-\lambda I)^{\dagger}G' \Pi (A+BK))'(G' \Pi G-\lambda I)q = 0$. Therefore, all terms involving $q$ vanish in the quadratic form $J'(G'\Pi G-\lambda I)J$.
Thus, substituting $B'\Pi GJ$ and $J$ in $\Pi(\lambda) = Q+A'\Pi A-K'B'\Pi BK -K'RK-K'B'\Pi GJ-J'G'\Pi BK-J'(G'\Pi G-\lambda I)J$ we obtain
\begin{align*}
\Pi(\lambda) =& Q+K'RK+(A+BK)'\Pi(A+BK)\\&-(A+BK)'\Pi G(G'\Pi G-\lambda I)^{\dagger}G' \Pi(A+BK)
\end{align*}
which is \eqref{xlrec} with $\bar{A} = A+BK$ and $\bar{Q} = Q + K'RK$.
\end{proof}

\subsection{\textbf{Proofs of Main Results}}
\label{app:proofs}

\begin{proof}[Proof of \cref{prop:2stage-StDAR}]
We proceed by backward dynamic programming. At each stage, we establish convexity and coercivity in the multipliers, concavity in the disturbance, and apply the minimax theorem (\cref{th:minimax}) together with pointwise strong duality (\cref{prop:conquad}) to interchange minimization and maximization. Throughout, \cref{receq} yields $\Pi_k(\cdot)\succeq0$ on its feasible domain. Define
\[
M_0(\lambda_0,\lambda_1)\eqbyd
\begin{bmatrix}
B'\Pi_1(\lambda_1)B+R & B'\Pi_1(\lambda_1)G\\
(B'\Pi_1(\lambda_1)G)' & G'\Pi_1(\lambda_1)G-\lambda_0 I
\end{bmatrix}
\]
\[
M_1(\lambda_1)\eqbyd
\begin{bmatrix}
B'P_fB+R & B'P_fG\\
(B'P_fG)' & G'P_fG-\lambda_1 I
\end{bmatrix}
\]

\textbf{First step: from $k=2$ to $k=1$.} \newline
We first show that the inequality constraint $|w_1|^2 \leq \alpha_1$ can be replaced by the equality constraint $|w_1|^2 = \alpha_1$. For fixed $(x_0,u_0,w_0,x_1,u_1)$, with $x_2 = Ax_1+Bu_1+Gw_1$, the numerator $V(x_0,\useq,\wseq)$ contains the term $(1/2)w_1'G'P_fGw_1 + w_1'G'P_f(Ax_1+Bu_1)$ in $w_1$, with the remaining terms independent of $w_1$. From $P_f \succeq 0$ we have $G'P_fG \succeq 0$, and \cref{asst3} gives $G'P_fG \neq 0$, so this term is convex but not constant in $w_1$. By \citet[Corollary~32.3.2]{rockafellar:1970}, a convex function on a compact convex set attains its maximum at an extreme point, so the maximum of $V$ over $\bbW_1 = \{w_1: |w_1|^2 \leq \alpha_1\}$ occurs on the boundary $|w_1|^2 = \alpha_1$. The inequality constraint $|w_1|^2 \leq \alpha_1$ can therefore be replaced by the equality constraint $|w_1|^2 = \alpha_1$ in the nested minmax of \eqref{2stagedp-w}.

With $|w_1|^2 = \alpha_1$ established, the problem becomes
\begin{align*}
&\min_{u_0}\max_{w_0 \in \bbW_0}\Bigg[\frac{1}{\overline{\alpha}}\Big(\ell(x_0,u_0) \\
&\quad +  \min_{u_1}\max_{w_1 \in \bbW_1}\big(\ell(x_1,u_1) + \ell_f(x_2) \big)\Big)\Bigg]
\end{align*}

Let
$\Lambda_1=\{\lambda_1:\lambda_1\ge\norm{G'P_fG}\}$, $h_1\eqbyd Ax_1+Bu_1$, and form the Lagrangian function
\[
L_1(u_1,w_1,\lambda_1)\eqbyd \ell(x_1,u_1)+\ell_f(h_1+Gw_1)-(\lambda_1/2)(\norm{w_1}^2-\alpha_1)
\]
By \cref{prop:conquad}, for every fixed $x_1$ and $u_1$, we have
\[
\max_{|w_1|^2=\alpha_1}\big[\ell(x_1,u_1)+\ell_f(h_1+Gw_1)\big] 
= \min_{\lambda_1\in\Lambda_1} \max_{w_1} L_1(u_1,w_1,\lambda_1)
\]
Hence
\[
\min_{u_1}\max_{w_1\in\bbW_1}\big[\ell(x_1,u_1)+\ell_f(x_2)\big] 
= \min_{u_1}\min_{\lambda_1\in\Lambda_1}\max_{w_1} L_1(u_1,w_1,\lambda_1)
\]
Since the two minimizations are over independent variables, they commute, and we obtain
\[
\min_{u_1}\min_{\lambda_1\in\Lambda_1}\max_{w_1} L_1 
= \min_{\lambda_1\in\Lambda_1}\min_{u_1}\max_{w_1} L_1
\]

Note that $M_1(\lambda_1)$ is invertible for all $\lambda_1\in\Lambda_1$ by \cref{prop:range-inclusion-invertible} and Assumptions 2-3. For any fixed $\lambda_1\in\Lambda_1$, \cref{prop:sdparalt} gives
\[
\min_{u_1}\max_{w_1} L_1(u_1,w_1,\lambda_1)
=\frac{1}{2}\,x_1'\Pi_1(\lambda_1)x_1+\frac{\alpha_1}{2}\lambda_1
\]
where 
\[
\Pi_1(\lambda_1)\eqbyd
Q+A'P_fA
- A'P_f\!\begin{bmatrix}B & G\end{bmatrix}
M_1(\lambda_1)^{-1}
\begin{bmatrix}B'\\ G'\end{bmatrix}\!P_f A
\]
From \cref{receq}, $\Pi_1(\lambda_1)$ can be rewritten as
$$\Pi_{1}(\lambda_1) = \bar{Q}_{1} + \bar{A}_{1}'P_f\bar{A}_{1} - \bar{A}_{1} P_f G(G'P_f G-\lambda_1 I)^{\dagger}G'P_f \bar{A}_{1}$$
where $\bar{A}_{1} = A+BK_{1}$ and $\bar{Q}_{1} = Q + K_{1}'RK_{1}$ and $K_{1}$ satisfies
    \begin{equation*}
  \begin{bmatrix}
        B'P_f B+R & B'P_f G \\
        G'P_f B & G'P_f G-\lambda_1 I
    \end{bmatrix}\begin{bmatrix}
K_{1} \\ J_{1}
\end{bmatrix} =-\begin{bmatrix}
B'P_f A \\ G'P_f A
\end{bmatrix}
\end{equation*}
From $Q\succeq0$ and $R\succ0$, we have $\bar{Q}_{1}\succeq0$. From $G'P_f G-\lambda_1 I \preceq0$, we have $\bar{A}_{1} P_f G(G'P_f G-\lambda_1 I)^{\dagger}G'P_f \bar{A}_{1} \preceq 0$. From $\bar{Q}_{1}\succeq0$, $G'P_f G-\lambda_1 I \preceq0$, $\bar{A}_{1} P_f G(G'P_f G-\lambda_1 I)^{\dagger}G'P_f \bar{A}_{1} \preceq 0$, and $P_f \succeq 0$, we have $\Pi_{1}(\lambda_1) \succeq0$ for $\lambda_1 \in \Lambda_1$.


Define $\phi_1(\lambda_1,x_1)\eqbyd \min_{u_1}\max_{w_1} L_1(u_1,w_1,\lambda_1)$. We establish joint convexity in $(\lambda_1,x_1)$ and coercivity in $\lambda_1$.

\emph{Convexity and coercivity of $\phi_1(\lambda_1,x_1)$.} 
For fixed $w_1$, the function $(u_1,\lambda_1,x_1) \mapsto L_1(u_1,w_1,\lambda_1)$ is convex: it is quadratic in $(x_1,u_1)$ with $Q\succeq0$, $R\succ0$, $P_f\succeq0$, and affine in $\lambda_1$. Define
\[
g_1(u_1,\lambda_1,x_1)\eqbyd \max_{w_1} L_1(u_1,w_1,\lambda_1)
\]
For $\lambda_1\in\Lambda_1$, the domain condition $G'P_fG - \lambda_1 I \preceq 0$ ensures $L_1$ is concave in $w_1$, hence the supremum over $w_1$ exists. Since $g_1$ is the pointwise supremum of convex functions in $(u_1,\lambda_1,x_1)$, it is jointly convex by \citet[§3.2.3]{boyd:vandenberghe:2004} (see also \citet[Theorem 5.5]{rockafellar:1970}). The partial minimization
\[
\phi_1(\lambda_1,x_1)\eqbyd \min_{u_1} g_1(u_1,\lambda_1,x_1)
\]
preserves joint convexity in $(\lambda_1,x_1)$ by \citet[§3.2.5]{boyd:vandenberghe:2004} (see also \citet[Theorem 5.3]{rockafellar:1970}).

For coercivity, note that $\phi_1(\lambda_1,x_1)=(1/2)x_1'\Pi_1(\lambda_1)x_1+(\alpha_1/2)\lambda_1$. Since $\Pi_1(\lambda_1)\succeq0$ and $\alpha_1>0$, we have $\phi_1(\lambda_1,x_1) \geq (\alpha_1/2)\lambda_1$, hence $\phi_1(\lambda_1,x_1) \to \infty$ as $\lambda_1 \to \infty$ on $\Lambda_1 = [\norm{G'P_fG},\infty)$.

\textbf{Second step: from $k=1$ to $k=0$.} \newline
With $|w_1|^2 = \alpha_1$ established from the first step, the remaining optimization is
\[
\min_{u_0}\max_{w_0 \in \bbW_0}\left[\frac{1}{\overline{\alpha}}\Big(\ell(x_0,u_0)+\min_{\lambda_1}\phi_1(\lambda_1,x_1)\Big)\right]
\]
where $x_1 = Ax_0+Bu_0+Gw_0$ and $\bbW_0 = \{w_0 \in \bbR^q : \norm{w_0}^2 \leq \alpha_0\}$. We similarly show that the inequality constraint $|w_0|^2 \leq \alpha_0$ can be replaced by the equality constraint $|w_0|^2 = \alpha_0$. For fixed $(x_0, u_0, \lambda_1)$ with $\lambda_1 \in \Lambda_1$, the inner cost $\ell(x_0,u_0) + \phi_1(\lambda_1,x_1)$ contains the term $(1/2)w_0'G'\Pi_1(\lambda_1)Gw_0 + w_0'G'\Pi_1(\lambda_1)(Ax_0+Bu_0)$ in $w_0$, with the remaining terms independent of $w_0$. Since $\Pi_1(\lambda_1) \succeq 0$ from the first step, we have $G'\Pi_1(\lambda_1)G \succeq 0$, so this term is convex in $w_0$. By \citet[Corollary~32.3.2]{rockafellar:1970}, the maximum over $\bbW_0$ occurs on the boundary $|w_0|^2 = \alpha_0$, and the inequality constraint can be replaced by the equality constraint. The same conclusion holds after the inner minimization over $\lambda_1$: by partial minimization \citep[§3.2.5]{boyd:vandenberghe:2004}, $\inf_{\lambda_1 \in \Lambda_1} \phi_1(\lambda_1, x_1)$ is convex in $x_1$, and composition with the affine mapping $x_1 = Ax_0 + Bu_0 + Gw_0$ preserves convexity in $w_0$.
Form the Lagrangian function
\begin{align*}
L_0(u_0,w_0,\lambda_0,\lambda_1) &\eqbyd \ell(x_0,u_0)
+\frac12 x_1'\Pi_1(\lambda_1)x_1\\
&\quad +\frac{\alpha_1}{2}\lambda_1
-\frac{\lambda_0}{2}(\norm{w_0}^2-\alpha_0)
\end{align*}
By \cref{prop:conlag}, for every fixed $x_0$, $u_0$, and $\lambda_1$, we have
\[
\max_{w_0\in\bbW_0}\big[\ell(x_0,u_0)+\phi_1(\lambda_1,x_1)\big] 
= \max_{w_0} \min_{\lambda_0} L_0(u_0,w_0,\lambda_0,\lambda_1)
\]
Hence
\begin{align*}
&\min_{u_0}\max_{w_0\in\bbW_0}\min_{\lambda_1}\big[\ell(x_0,u_0)+\phi_1(\lambda_1,x_1)\big]\\
&\quad = \min_{u_0} \max_{w_0} \min_{\lambda_0,\lambda_1} L_0(u_0,w_0,\lambda_0,\lambda_1)
\end{align*}

We justify the interchange of minimization ($\min_{\lambda_0,\lambda_1}$) and maximization ($\max_{w_0}$). The feasible set for $(\lambda_0,\lambda_1)$ is
\[
\Lambda_2 \eqbyd \{(\lambda_0,\lambda_1): \lambda_1 \geq \norm{G'P_fG}, \lambda_0 \geq \norm{G'\Pi_1(\lambda_1)G}\}
\]
which is not a product set since the lower bound on $\lambda_0$ depends on $\lambda_1$. We show $\Lambda_2$ is convex. From the closed form $\phi_1(\lambda_1, x_1) = (1/2)x_1'\Pi_1(\lambda_1)x_1 + (\alpha_1/2)\lambda_1$ and joint convexity of $\phi_1$ in $(\lambda_1, x_1)$ established above, the map $(\lambda_1, x_1) \mapsto x_1'\Pi_1(\lambda_1)x_1$ is jointly convex on $\Lambda_1 \times \bbR^n$ (subtracting the term affine in $\lambda_1$ preserves convexity). For each $y \in \bbR^q$, restricting to $x_1 = Gy$ gives $\lambda_1 \mapsto y'G'\Pi_1(\lambda_1)Gy$ convex on $\Lambda_1$. Since $\Pi_1(\lambda_1) \succeq 0$, we have $G'\Pi_1(\lambda_1)G \succeq 0$ and the induced 2-norm of a positive semidefinite matrix equals its largest eigenvalue, giving
\[
\norm{G'\Pi_1(\lambda_1)G} = \sup_{|y|=1} y'G'\Pi_1(\lambda_1)Gy
\]
The function $\lambda_1 \mapsto \norm{G'\Pi_1(\lambda_1)G}$ is therefore convex on $\Lambda_1$ as the pointwise supremum of convex functions \citep[§3.2.3]{boyd:vandenberghe:2004}. The set $\Lambda_2$ is the intersection of $\Lambda_1 \times \bbR$ and $\{(\lambda_0, \lambda_1) : \lambda_0 \geq \norm{G'\Pi_1(\lambda_1)G}\}$. The first set is convex because $\Lambda_1 = [\norm{G'P_fG}, \infty)$ is convex; the second is convex by \citet[§3.1.7]{boyd:vandenberghe:2004} since $\norm{G'\Pi_1(\lambda_1)G}$ is convex in $\lambda_1$. Hence $\Lambda_2$ is convex as the intersection of two convex sets. We proceed by first eliminating $\lambda_0$.

\emph{Step 2a: Minimax interchange between $w_0$ and $\lambda_1$.}
Define the reduced function
\[
\psi(u_0, w_0, \lambda_1) \eqbyd \inf_{\lambda_0 \geq \norm{G'\Pi_1(\lambda_1)G}} L_0(u_0,w_0,\lambda_0,\lambda_1)
\]
Since $L_0$ is affine in $\lambda_0$ with coefficient $-\frac{1}{2}(\norm{w_0}^2 - \alpha_0)$, the function $\psi$ is extended-valued: $\psi(u_0,w_0,\lambda_1) = -\infty$ for $\norm{w_0}^2 > \alpha_0$, and $\psi$ is finite for $\norm{w_0}^2 \leq \alpha_0$ with the infimum attained at $\lambda_0 = \norm{G'\Pi_1(\lambda_1)G}$ for $\norm{w_0}^2 < \alpha_0$ and at any $\lambda_0 \geq \norm{G'\Pi_1(\lambda_1)G}$ for $\norm{w_0}^2 = \alpha_0$. Therefore
\begin{align*}
\max_{w_0} \min_{\lambda_0,\lambda_1} L_0 &= \max_{w_0} \min_{\lambda_1 \in \Lambda_1} \psi(u_0,w_0,\lambda_1) \\
&= \max_{\norm{w_0}^2 \leq \alpha_0} \min_{\lambda_1 \in \Lambda_1} \psi(u_0,w_0,\lambda_1)
\end{align*}
The function $\psi$ is concave in $w_0$ since it is the pointwise infimum of functions concave in $w_0$ by \citet[§3.2.3]{boyd:vandenberghe:2004}. For fixed $(u_0, w_0)$, since $L_0 = \ell(x_0,u_0) + \phi_1(\lambda_1,x_1) - (\lambda_0/2)(\norm{w_0}^2-\alpha_0)$, $\phi_1(\lambda_1,x_1)$ is convex in $\lambda_1$ from the first step, and the $\lambda_0$ term is independent of $\lambda_1$, the function $L_0$ is convex in $\lambda_1$. The partial minimization over $\lambda_0$ preserves convexity in $\lambda_1$ by \citet[§3.2.5]{boyd:vandenberghe:2004}, hence $\psi$ is convex in $\lambda_1$.

We verify the hypotheses of \cref{th:minimax}. The effective domain in $w_0$ is $\{w_0: \norm{w_0}^2 \leq \alpha_0\}$, and we have established that the maximum occurs on the boundary $\norm{w_0}^2 = \alpha_0$, which is compact. The coercivity $\psi(u_0,w_0,\lambda_1) \geq (\alpha_1/2)\lambda_1$ (since $Q \succeq 0$, $R \succ 0$, and $\Pi_1(\lambda_1) \succeq 0$) allows restriction of $\lambda_1$ to a compact interval $[\norm{G'P_fG}, \bar{\lambda}_1]$ containing all minimizers. The function $\psi$ is continuous on this domain. Therefore \cref{th:minimax} yields
\[
\max_{w_0} \min_{\lambda_1 \in \Lambda_1} \psi(u_0,w_0,\lambda_1) = \min_{\lambda_1 \in \Lambda_1} \max_{w_0} \psi(u_0,w_0,\lambda_1)
\]
Translating back, we obtain
\[
\max_{w_0} \min_{\lambda_0,\lambda_1} L_0 = \min_{\lambda_1} \max_{w_0} \min_{\lambda_0 \geq \norm{G'\Pi_1(\lambda_1)G}} L_0
\]

\emph{Step 2b: Strong duality between $(w_0, \lambda_0)$ for fixed $\lambda_1$.}
For each fixed $\lambda_1 \in \Lambda_1$, the inner optimization $\max_{w_0} \min_{\lambda_0 \geq \norm{G'\Pi_1(\lambda_1)G}} L_0$ is a sphere constrained quadratic maximization in $w_0$ with Lagrange multiplier $\lambda_0$. By \cref{prop:conquad}, strong duality holds, and we obtain
\begin{align*}
&\max_{w_0} \min_{\lambda_0 \geq \norm{G'\Pi_1(\lambda_1)G}} L_0(u_0,w_0,\lambda_0,\lambda_1) \\
&\quad = \min_{\lambda_0 \geq \norm{G'\Pi_1(\lambda_1)G}} \max_{w_0} L_0(u_0,w_0,\lambda_0,\lambda_1)
\end{align*}
This equality holds pointwise for every $\lambda_1 \in \Lambda_1$. Substituting inside the outer minimization over $\lambda_1$, we obtain
\[
\min_{\lambda_1} \max_{w_0} \min_{\lambda_0} L_0 = \min_{\lambda_1} \min_{\lambda_0} \max_{w_0} L_0 = \min_{(\lambda_0,\lambda_1) \in \Lambda_2} \max_{w_0} L_0
\]
Therefore
\begin{align*}
\min_{u_0}\max_{w_0}\min_{(\lambda_0,\lambda_1)\in\Lambda_2} L_0 &= \min_{u_0}\min_{(\lambda_0,\lambda_1)\in\Lambda_2}\max_{w_0} L_0 \\
&= \min_{(\lambda_0,\lambda_1)\in\Lambda_2}\min_{u_0}\max_{w_0} L_0
\end{align*}

Note that $M_0(\lambda_0,\lambda_1)$ is invertible for all $(\lambda_0,\lambda_1) \in \Lambda_2$ by \cref{prop:range-inclusion-invertible} and Assumptions 2-3. For fixed $(\lambda_0,\lambda_1)\in\Lambda_2$, \cref{prop:sdparalt} gives
\[
\min_{u_0}\max_{w_0} L_0
=\frac{1}{2}\,x_0'\Pi_0(\lambda_0,\lambda_1)x_0+\frac{\alpha_0}{2}\lambda_0+\frac{\alpha_1}{2}\lambda_1
\]
where
\begin{align*}
\Pi_0(\lambda_0,\lambda_1) &\eqbyd Q+A'\Pi_1(\lambda_1)A\\
&\quad - A'\Pi_1(\lambda_1)\!\begin{bmatrix}B & G\end{bmatrix}
M_0(\lambda_0,\lambda_1)^{-1}\\
&\quad \quad \begin{bmatrix}B'\\ G'\end{bmatrix}\!\Pi_1(\lambda_1) A
\end{align*}
From the same arguments that proved $\Pi_1(\lambda_1)\succeq0$ for $\lambda_1\in \Lambda_1$, we have $\Pi_0(\lambda_0,\lambda_1)\succeq0$ for $(\lambda_0,\lambda_1)\in\Lambda_2$. 

Define $\phi_0(\lambda_0,\lambda_1,x_0)\eqbyd \min_{u_0}\max_{w_0} L_0(u_0,w_0,\lambda_0,\lambda_1)$. We establish joint convexity in $(\lambda_0,\lambda_1,x_0)$ and coercivity in $(\lambda_0,\lambda_1)$.

\emph{Convexity and coercivity of $\phi_0(\lambda_0,\lambda_1,x_0)$.}
For fixed $w_0$, the function $(u_0,\lambda_0,\lambda_1,x_0) \mapsto L_0(u_0,w_0,\lambda_0,\lambda_1)$ is convex. Since $L_0 = \ell(x_0,u_0) + \phi_1(\lambda_1,x_1) - (\lambda_0/2)(\norm{w_0}^2-\alpha_0)$ with $x_1 = Ax_0+Bu_0+Gw_0$, we have: $\ell(x_0,u_0)$ is quadratic with $Q\succeq0$ and $R\succ0$, hence convex in $(x_0,u_0)$; the term $-(\lambda_0/2)(\norm{w_0}^2-\alpha_0)$ is affine in $\lambda_0$ for fixed $w_0$; and $\phi_1(\lambda_1,x_1)$ is jointly convex in $(\lambda_1,x_1)$ from the first step. Since $x_1$ is affine in $(x_0,u_0,w_0)$, the composition $\phi_1(\lambda_1,Ax_0+Bu_0+Gw_0)$ is jointly convex in $(\lambda_1,x_0,u_0,w_0)$ by \citet[§3.2.2]{boyd:vandenberghe:2004}. Hence $L_0$ is jointly convex in $(u_0,\lambda_0,\lambda_1,x_0)$ for fixed $w_0$. Define
\[
g_0(u_0,\lambda_0,\lambda_1,x_0)\eqbyd \max_{w_0} L_0(u_0,w_0,\lambda_0,\lambda_1)
\]
For $(\lambda_0,\lambda_1)\in\Lambda_2$, the domain condition $\lambda_0 I - G'\Pi_1(\lambda_1)G \succeq 0$ ensures $L_0$ is concave in $w_0$, hence the supremum over $w_0$ exists. Since $g_0$ is the pointwise supremum of convex functions in $(u_0,\lambda_0,\lambda_1,x_0)$, it is jointly convex by \citet[§3.2.3]{boyd:vandenberghe:2004} (see also \citet[Theorem 5.5]{rockafellar:1970}). The partial minimization
\[
\phi_0(\lambda_0,\lambda_1,x_0)\eqbyd \min_{u_0} g_0(u_0,\lambda_0,\lambda_1,x_0)
\]
preserves joint convexity in $(\lambda_0,\lambda_1,x_0)$ by \citet[§3.2.5]{boyd:vandenberghe:2004} (see also \citet[Theorem 5.3]{rockafellar:1970}).

For coercivity, note that $\phi_0(\lambda_0,\lambda_1,x_0)=(1/2)x_0'\Pi_0(\lambda_0,\lambda_1)x_0+(\alpha_0/2)\lambda_0+(\alpha_1/2)\lambda_1$. Since $\Pi_0(\lambda_0,\lambda_1)\succeq0$ and $\alpha_0, \alpha_1>0$, we have $\phi_0(\lambda_0,\lambda_1,x_0) \geq (\alpha_0/2)\lambda_0 + (\alpha_1/2)\lambda_1$, hence $\phi_0$ is coercive in $(\lambda_0,\lambda_1)$ on $\Lambda_2$.

\textbf{Third step: final optimization.} \newline
Finally, the outer minimization over $(\lambda_0,\lambda_1)\in\Lambda_2$ yields
\begin{align*}
&\min_{(\lambda_0,\lambda_1)\in\Lambda_2}\  \phi_0(\lambda_0,\lambda_1,x_0)\\
&\quad =\min_{(\lambda_0,\lambda_1)\in\Lambda_2}\ 
\frac{1}{2\overline{\alpha}}\,x_0'\,\Pi_0(\lambda_0,\lambda_1)\,x_0 +\frac{1}{2\overline{\alpha}}\big(\alpha_0\lambda_0+\alpha_1\lambda_1\big)
\end{align*}
which is the optimization \eqref{eq:2stage-opt}. Existence of the optimal solution $(\lambda_0^*,\lambda_1^*)$ follows from the Weierstrass theorem: coercivity holds by the analysis in the second step, continuity holds because $M_0(\lambda_0,\lambda_1)$ is invertible for all $(\lambda_0,\lambda_1) \in \Lambda_2$ by \cref{prop:range-inclusion-invertible} and Assumptions 2-3, hence $\Pi_0(\lambda_0,\lambda_1)$ is continuous on $\Lambda_2$ as a composition of continuous operations, and $\Lambda_2$ is closed. 

Given the solution $(\lambda_0^*,\lambda_1^*)$, we have
\begin{align*}
L^*_0(\lambda_0^*,\lambda_1^*) &= V(x_0, \useq^*, \wseq^*) \\
&\quad - \frac{\lambda_0^*}{2}(|w_0^*|^2 - \alpha_0) - \frac{\lambda_1^*}{2}(|w_1^*|^2 - \alpha_1) \\
&= \frac{1}{2}x_0'\Pi_0(\lambda_0^*,\lambda_1^*)x_0 + \frac{\alpha_0\lambda_0^* + \alpha_1\lambda_1^*}{2}
\end{align*}
and since $\wseq^*$ satisfies $|w_0^*|^2 = \alpha_0$ and $|w_1^*|^2 = \alpha_1$ from the constraints $w_0^* \in \bbW_0$ and $w_1^* \in \bbW_1$, we obtain
\[
V^*(x_0) = \min_{u_0}\max_{w_0 \in \bbW_0}\min_{u_1}\max_{w_1 \in \bbW_1} \frac{V(x_0, \useq, \wseq)}{\overline{\alpha}} = \frac{L^*_0(\lambda_0^*,\lambda_1^*)}{\overline{\alpha}}
\]
which is \eqref{2stage-cost}. 

Given $(\lambda_0^*,\lambda_1^*)$, items 1-2 follow from the closed form solutions at each stage via \cref{prop:sdparalt}. Item 4 follows from \cref{receq}.
\end{proof}

\begin{proof}[Proof of \cref{prop:ndstage}]
The proof follows by induction from \cref{prop:2stage-StDAR}.
Define
\begin{align*}
M_k(\boldsymbol{\lambda}_k) &\eqbyd \begin{bmatrix}
B'\Pi_{k+1}B+R & B'\Pi_{k+1}G \\
G'\Pi_{k+1}B & G'\Pi_{k+1}G-\lambda_k I
\end{bmatrix} \\
d_k &\eqbyd \begin{bmatrix}
B'\Pi_{k+1}A \\ G'\Pi_{k+1}A
\end{bmatrix}x_k
\end{align*}
where $\boldsymbol{\lambda}_k = (\lambda_k,\ldots,\lambda_{N-1})$ denotes the vector of multipliers from stage $k$ onward. We apply backward dynamic programming, solving each minmax subproblem at stage $k$ using the minimax theorem (\cref{th:minimax}) together with pointwise strong duality (\cref{prop:conquad}).

\textbf{Terminal stage $k=N-1$.} \newline
At the terminal stage, the analysis from the first step of \cref{prop:2stage-StDAR} applies directly with $\Pi_N = P_f$, yielding
\[
\phi_{N-1}(\lambda_{N-1},x_{N-1})=\frac{1}{2}\,x_{N-1}'\Pi_{N-1}(\lambda_{N-1})x_{N-1}+\frac{\alpha_{N-1}}{2}\lambda_{N-1}
\]
where $\Pi_{N-1}(\lambda_{N-1})$ is given by \eqref{stagerec1}. From \cref{receq}, $\Pi_{N-1}(\lambda_{N-1})\succeq0$ for $\lambda_{N-1} \geq \norm{G'P_fG}$. The set $\Lambda_1 = [\norm{G'P_fG},\infty)$ is convex. By the convexity and coercivity arguments in \cref{prop:2stage-StDAR}, $\phi_{N-1}(\lambda_{N-1},x_{N-1})$ is jointly convex in $(\lambda_{N-1},x_{N-1})$ and coercive in $\lambda_{N-1}$ on $\Lambda_1$.

\textbf{Inductive step: stage $k \in \{0,\ldots,N-2\}$.} \newline
At each stage $k \in \{0,\ldots,N-2\}$, assume by induction that $\Lambda_{N-k-1}$ is convex, $\phi_{k+1}(\boldsymbol{\lambda}_{k+1},x_{k+1})$ is jointly convex in $(\boldsymbol{\lambda}_{k+1},x_{k+1})$ and coercive in $\boldsymbol{\lambda}_{k+1}$ on $\Lambda_{N-k-1}$, and $\Pi_{k+1}(\boldsymbol{\lambda}_{k+1})\succeq0$ for all $\boldsymbol{\lambda}_{k+1}\in\Lambda_{N-k-1}$.

By the argument from \cref{prop:2stage-StDAR}, the inequality constraint $|w_k|^2 \leq \alpha_k$ can be replaced by the equality constraint $|w_k|^2 = \alpha_k$. For fixed $(x_k, u_k, \boldsymbol{\lambda}_{k+1})$ with $\boldsymbol{\lambda}_{k+1} \in \Lambda_{N-k-1}$ and $x_{k+1} = Ax_k+Bu_k+Gw_k$, the inner cost $\ell(x_k,u_k) + \phi_{k+1}(\boldsymbol{\lambda}_{k+1}, x_{k+1})$ contains the term $(1/2)w_k'G'\Pi_{k+1}(\boldsymbol{\lambda}_{k+1})Gw_k + w_k'G'\Pi_{k+1}(\boldsymbol{\lambda}_{k+1})(Ax_k+Bu_k)$ in $w_k$ (from $(1/2)x_{k+1}'\Pi_{k+1}(\boldsymbol{\lambda}_{k+1})x_{k+1}$), with the remaining terms independent of $w_k$. Since $\Pi_{k+1}(\boldsymbol{\lambda}_{k+1}) \succeq 0$ by induction, we have $G'\Pi_{k+1}G \succeq 0$ and this term is convex in $w_k$. By \citet[Corollary~32.3.2]{rockafellar:1970}, the maximum over $\{w_k: \norm{w_k}^2 \leq \alpha_k\}$ occurs on the boundary $\norm{w_k}^2 = \alpha_k$. The same conclusion holds after the inner minimization over $\boldsymbol{\lambda}_{k+1}$: by partial minimization \citep[§3.2.5]{boyd:vandenberghe:2004}, $\inf_{\boldsymbol{\lambda}_{k+1} \in \Lambda_{N-k-1}} \phi_{k+1}(\boldsymbol{\lambda}_{k+1}, x_{k+1})$ is convex in $x_{k+1}$, and composition with the affine mapping $x_{k+1} = Ax_k + Bu_k + Gw_k$ preserves convexity in $w_k$.

Form the Lagrangian function
\begin{align*}
L_k(u_k,w_k,\lambda_k,\boldsymbol{\lambda}_{k+1}) &\eqbyd \ell(x_k,u_k) +\frac12 x_{k+1}'\Pi_{k+1}(\boldsymbol{\lambda}_{k+1})x_{k+1}\\
&\quad +\frac{1}{2}\sum_{j=k+1}^{N-1}\alpha_j\lambda_j -\frac{\lambda_k}{2}(\norm{w_k}^2-\alpha_k)
\end{align*}
By \cref{prop:conlag}, the equality constrained problem equals the Lagrangian saddle point problem. Define
\[
\Lambda_{N-k} \eqbyd \{(\lambda_k,\boldsymbol{\lambda}_{k+1}): \boldsymbol{\lambda}_{k+1}\in\Lambda_{N-k-1}, \lambda_k \geq \norm{G'\Pi_{k+1}(\boldsymbol{\lambda}_{k+1})G}\}
\]
which is not a product set since the lower bound on $\lambda_k$ depends on $\boldsymbol{\lambda}_{k+1}$. We show $\Lambda_{N-k}$ is convex. From the closed form $\phi_{k+1}(\boldsymbol{\lambda}_{k+1}, x_{k+1}) = (1/2)x_{k+1}'\Pi_{k+1}(\boldsymbol{\lambda}_{k+1})x_{k+1} + (1/2)\sum_{j=k+1}^{N-1}\alpha_j\lambda_j$ and joint convexity of $\phi_{k+1}$ in $(\boldsymbol{\lambda}_{k+1}, x_{k+1})$ from the induction hypothesis, the map $(\boldsymbol{\lambda}_{k+1}, x_{k+1}) \mapsto x_{k+1}'\Pi_{k+1}(\boldsymbol{\lambda}_{k+1})x_{k+1}$ is jointly convex on $\Lambda_{N-k-1} \times \bbR^n$ (subtracting the term affine in $\boldsymbol{\lambda}_{k+1}$ preserves convexity). For each $y \in \bbR^q$, restricting to $x_{k+1} = Gy$ gives $\boldsymbol{\lambda}_{k+1} \mapsto y'G'\Pi_{k+1}(\boldsymbol{\lambda}_{k+1})Gy$ convex on $\Lambda_{N-k-1}$. Since $\Pi_{k+1}(\boldsymbol{\lambda}_{k+1}) \succeq 0$, we have $G'\Pi_{k+1}(\boldsymbol{\lambda}_{k+1})G \succeq 0$ and the induced 2-norm of a positive semidefinite matrix equals its largest eigenvalue, giving
\[
\norm{G'\Pi_{k+1}(\boldsymbol{\lambda}_{k+1})G} = \sup_{|y|=1} y'G'\Pi_{k+1}(\boldsymbol{\lambda}_{k+1})Gy
\]
The function $\boldsymbol{\lambda}_{k+1} \mapsto \norm{G'\Pi_{k+1}(\boldsymbol{\lambda}_{k+1})G}$ is therefore convex on $\Lambda_{N-k-1}$ as the pointwise supremum of convex functions \citep[§3.2.3]{boyd:vandenberghe:2004}. The set $\Lambda_{N-k}$ is the intersection of $\Lambda_{N-k-1} \times \bbR$ and $\{(\lambda_k, \boldsymbol{\lambda}_{k+1}) : \lambda_k \geq \norm{G'\Pi_{k+1}(\boldsymbol{\lambda}_{k+1})G}\}$. The first set is convex by the induction hypothesis; the second is convex by \citet[§3.1.7]{boyd:vandenberghe:2004} since $\norm{G'\Pi_{k+1}(\boldsymbol{\lambda}_{k+1})G}$ is convex in $\boldsymbol{\lambda}_{k+1}$. Hence $\Lambda_{N-k}$ is convex as the intersection of two convex sets. By the arguments in the second step of \cref{prop:2stage-StDAR}, we first eliminate $\lambda_k$ by defining the reduced function
\[
\psi_k(u_k, w_k, \boldsymbol{\lambda}_{k+1}) \eqbyd \inf_{\lambda_k \geq \norm{G'\Pi_{k+1}(\boldsymbol{\lambda}_{k+1})G}} L_k(u_k,w_k,\lambda_k,\boldsymbol{\lambda}_{k+1})
\]
The function $\psi_k$ is extended-valued with $\psi_k = -\infty$ for $\norm{w_k}^2 > \alpha_k$ and finite for $\norm{w_k}^2 \leq \alpha_k$. By \citet[§3.2.3]{boyd:vandenberghe:2004}, $\psi_k$ is concave in $w_k$ and convex in $\boldsymbol{\lambda}_{k+1}$. The minimax theorem (\cref{th:minimax}) yields the interchange between $w_k$ and $\boldsymbol{\lambda}_{k+1}$, and pointwise strong duality (\cref{prop:conquad}) for each fixed $\boldsymbol{\lambda}_{k+1}$ yields the interchange between $w_k$ and $\lambda_k$. Therefore
\begin{align*}
\min_{u_k}\max_{w_k}\min_{(\lambda_k,\boldsymbol{\lambda}_{k+1})\in\Lambda_{N-k}} L_k &= \min_{u_k}\min_{(\lambda_k,\boldsymbol{\lambda}_{k+1})\in\Lambda_{N-k}}\max_{w_k} L_k \\
&= \min_{(\lambda_k,\boldsymbol{\lambda}_{k+1})\in\Lambda_{N-k}} \min_{u_k}\max_{w_k} L_k
\end{align*}

Note that $M_k(\boldsymbol{\lambda}_k)$ is invertible for all $(\lambda_k,\boldsymbol{\lambda}_{k+1}) \in \Lambda_{N-k}$ by \cref{prop:range-inclusion-invertible} and Assumptions 2-3. For fixed $(\lambda_k,\boldsymbol{\lambda}_{k+1})\in\Lambda_{N-k}$, \cref{prop:sdparalt} gives
\[
\min_{u_k}\max_{w_k} L_k
=\frac{1}{2}\,x_k'\Pi_k(\boldsymbol{\lambda}_k)x_k+\frac{1}{2}\sum_{j=k}^{N-1}\alpha_j\lambda_j
\]
where $\Pi_k(\boldsymbol{\lambda}_k)$ is given by \eqref{stagerec1}. From the same arguments that proved $\Pi_{k+1}(\boldsymbol{\lambda}_{k+1})\succeq0$, we have $\Pi_k(\boldsymbol{\lambda}_k)\succeq0$ for $\boldsymbol{\lambda}_k\in\Lambda_{N-k}$. 

Define $\phi_k(\boldsymbol{\lambda}_k,x_k)\eqbyd \min_{u_k}\max_{w_k} L_k(u_k,w_k,\lambda_k,\boldsymbol{\lambda}_{k+1})$. By the convexity arguments from \cref{prop:2stage-StDAR}, $\phi_k(\boldsymbol{\lambda}_k,x_k)$ is jointly convex in $(\boldsymbol{\lambda}_k,x_k)$. For coercivity, $\phi_k(\boldsymbol{\lambda}_k,x_k) \geq (1/2)\sum_{j=k}^{N-1}\alpha_j\lambda_j$, hence $\phi_k$ is coercive in $\boldsymbol{\lambda}_k$ on $\Lambda_{N-k}$.

\textbf{Final optimization.} \newline
By induction, we obtain the recursion \eqref{stagerec1} for $k \in \{0,1,\ldots,N-1\}$ with terminal condition $\Pi_N = P_f$, and the remaining optimization is \eqref{lst} with $\boldsymbol{\lambda}_0$ ranging over $\Lambda_N$. The set $\Lambda_N$ is convex by induction and the objective $\phi_0(\boldsymbol{\lambda}_0, x_0)$ is jointly convex in $(\boldsymbol{\lambda}_0, x_0)$, so \eqref{lst} is a convex optimization. Existence of the optimal solution $\boldsymbol{\lambda}_0^* = (\lambda_0^*,\ldots,\lambda_{N-1}^*)$ follows from the Weierstrass theorem: coercivity holds by the analysis above, continuity holds because $M_k(\boldsymbol{\lambda}_k)$ is invertible for all $\boldsymbol{\lambda}_k \in \Lambda_{N-k}$ by \cref{prop:range-inclusion-invertible} and Assumptions 2-3, hence $\Pi_k(\boldsymbol{\lambda}_k)$ is continuous on $\Lambda_{N-k}$ as a composition of continuous operations, and $\Lambda_N$ is closed.

Given $\boldsymbol{\lambda}_0^*$, items 1--3 follow from the closed form solutions at each stage via \cref{prop:sdparalt}. Item 4 follows from \cref{receq}.
\end{proof}

\bibliographystyle{abbrvnat}
\bibliography{paper_arxiv}

\end{document}